\documentclass[11pt]{article}

\usepackage[preprint]{acl}

\usepackage{times}
\usepackage{latexsym}
\usepackage{multirow}
\usepackage[T1]{fontenc}

\usepackage[utf8]{inputenc}

\usepackage{microtype}

\usepackage{inconsolata}

\usepackage{graphicx}

\usepackage{enumitem}
\usepackage{amsmath}
\usepackage{amssymb}
\usepackage{booktabs}
\usepackage{enumitem}
\usepackage{xspace}

\usepackage[skins]{tcolorbox}
\usepackage{tabularx}
\usepackage{CJKutf8}
\definecolor{realpii}{rgb}{0.75,0.10,0.10}
\definecolor{surrogate}{rgb}{0.05,0.30,0.70}
\definecolor{residue}{rgb}{0.85,0.45,0.00}
\usepackage{amsthm}
\newtheorem{theorem}{Theorem}
\newtheorem{corollary}{Corollary}
\newcommand{\Adv}{\mathrm{Adv}}
\newcommand{\Leak}{\mathcal{L}}
\newcommand{\San}{\mathsf{San}}

\newcommand{\sys}{\textsc{VtR}\xspace}

\newcommand{\Tstr}{\mathcal{T}_{\mathrm{syn}}}
\newcommand{\Tsem}{\mathcal{T}_{\mathrm{sem}}}
\newcommand{\Ks}{K_{s}}
\newcommand{\Ksyn}{K_{\mathrm{syn}}}
\newcommand{\Km}{K_{\mathrm{master}}}

\title{
Veil-Then-Reveal: Protecting PII Privacy via Reversible Pseudo-Anonymization in Zero-Trust User--LLM Interaction
}

\author{Mohamed Elmahallawy$^\dagger$, Mohamed Shaaban$^\dagger$, and Thomas E. Carroll$^*$ \\
  $^\dagger$School of Engineering and Applied Sciences, Washington State University, Richland, WA 99354, USA \\
  $^*$Pacific Northwest National Laboratory, Richland, WA 99354, USA\\
  \texttt{\{mohamed.elmahallawy,mohamed.shaaban\}@wsu.edu, thomas.carroll@pnnl.gov}
  }

\usepackage[normalem]{ulem}                        

\newcommand{\sref}[1]{Section~\ref{#1}}
\usepackage[misc]{ifsym}
\newcommand{\etal}{et al.}

\begin{document}
\maketitle
\begin{abstract}
Large language models (LLMs) increasingly process sensitive documents and personal information, yet sending prompts containing personally identifiable information (PII) to remote providers exposes {\em data subjects} and {\em organizations} to risks from logging, retention, model training, and unintended disclosure. Existing anonymization methods often {\em mask or remove} PII, while conventional encryption eliminates the lexical and structural cues needed for effective LLM processing. We present {\em Veil-then-Reveal (\sc VtR)}, a user-side PII pseudo-anonymization framework for privacy-preserving LLM interaction under a zero-trust threat model. {\sc VtR} treats remote LLM providers as untrusted and sends only pseudo-anonymized prompts, keeping plaintext PII on the user. \sys first detects PII and replaces it with realistic, same-type, in-vocabulary surrogates before prompts leave the user's device, preserving their linguistic and structural properties. It combines {\em format-preserving encryption (FPE)} for structured PII, such as phone and credit-card numbers, with a {\em keyed-hash dictionary and secure vault} for semantic PII, such as names, locations, and organizations. Session-specific mappings provide cross-session unlinkability, while user-side restoration reversibly replaces surrogates with their original values in LLM responses. We further introduce a rank-class FPE strategy for dates that preserves format and ambiguity while reducing cipher evaluations. Evaluation on three datasets containing 11,322 annotated PII instances achieves {97.9\%} detection recall, {95.2\%} protection of detected values, and {99.9\% round trip; 100\% of echoed surrogates} end-to-end round-trip restoration. 
\end{abstract}

\section{Introduction}
Large language models (LLMs) increasingly process sensitive documents containing {\em personally identifiable information (PII)}. Sending raw PII to remote LLM providers creates substantial privacy risks: prompts may be logged, retained, exposed through model outputs, or exploited to recover sensitive information \cite{das2025security,cheng2026privacy}. Recent incidents further demonstrate that even widely deployed LLM services can expose user data \cite{openai2023chatgptoutage}. These risks motivate a fundamental question: {\em how can users leverage remote LLMs without revealing their raw PII?}

Existing defenses largely rely on redaction or anonymization, but removing PII can destroy context-critical information needed for accurate inference \cite{asthana2025adaptive}. Simple pseudonymization can also leave indirect identifiers that enable re-identification \cite{boutet2025towards}, while masking alone does not prevent sensitive information from resurfacing through retrieval or memory mechanisms \cite{ma2026bodhipromptshield}. Recent approaches such as \emph{PrivacyRestore} \cite{zeng2025privacyrestore} and \emph{SecureGate} \cite{shaaban2026securegate} explore reversible or selective protection, but introduce additional model-side mechanisms, training, or communication overhead.

We present {\em Veil-then-Reveal (\sc{VtR})}, a user-side PII pseudo-anonymization framework that enables user-LLM interaction under a {\em zero-trust} threat model. Instead of deleting sensitive values, {\sc VtR} first detects PII and then replaces each instance locally with realistic, same-type surrogates, allowing the remote LLM to operate on structurally and semantically meaningful text while never receiving the original PII. {\sc VtR} integrates {\em Format-Preserving Encryption (FPE)} for structured identifiers with {\em keyed semantic substitution} for contextual entities. Original PII values and mappings remain exclusively on the user, enabling exact restoration after inference and session-specific unlinkability.

A key novelty is our {\em rank-class FPE} for dates. Unlike digit-wise encryption, which can produce invalid dates or destroy inherent date ambiguity, rank-class FPE encrypts a date within the set of legal interpretations admitted by its rendering. This preserves both syntactic validity and ambiguity while reducing the encryption overhead by {$410$--$730\times$}. In summary, our key contributions are:
\begin{itemize}[leftmargin=*]
\item We propose \emph{\sc{VtR}}, a reversible, user-side pseudo-anonymization framework that keeps raw PII outside the boundary of a remote LLM while preserving useful prompt context and model utility.

\item We develop a {\em hybrid protection design} combining FPE with keyed semantic token substitution, using user-only mappings and session-specific surrogates to provide cross-session unlinkability.

\item We propose {\em rank-class FPE for dates}, retaining valid formatting and temporal ambiguity while markedly reducing encryption cost and improving valid-date generation from {$0.3\%$ to $88.4\%$.}

\item We evaluate {\sc VtR} on {three datasets containing}~{four corpora, including 10,000 ai4privacy documents with} 11,322 annotated PII instances across 8 PII types, achieving {97.9\%} detection recall, {95.2\%} protection of detected values, and {99.9\%} exact round-trip restoration. The results show that residual exposure is governed by detection recall rather than transformation failures.
\end{itemize}

\section{Related Work}

LLMs pose privacy risks through prompt exposure and context leakage \cite{mireshghallah2025position}; multi-turn interactions can expose over $86\%$ of sensitive content \cite{agarwal2024prompt}. Existing defenses use {\em PII detection, cryptographic sanitization, and user-side protection}, but often trade privacy for utility or reversibility.

\vspace{0.7mm}
\noindent{\bf \em PII Detection and Text Anonymization.}
PII protection typically begins with entity detection and masking. Savkin \etal~\cite{savkin2025spy} show that generic NER models often miss PII-specific contexts, while RECAP~\cite{rajgarhia2025evaluation} improves recall by combining regex and LLM-based detection. Undetected PII remains exposed, making detection recall a fundamental privacy limitation. Conventional anonymization, however, masks or removes PII, disrupting linguistic structure and preventing exact recovery. DP-based approaches such as {\sc CusText}~\cite{chen2023customized} provide formal privacy guarantees via token perturbation but may reduce utility and lack exact restoration.

\vspace{1mm}
\noindent{\bf \em Cryptographic Sanitization.}
FPE, such as NIST FF3-1, protects structured values while preserving their format and length. \textsc{Pr}\ensuremath{\epsilon\epsilon}\textsc{mpt}~\cite{chowdhury2025pr} combines FPE for fixed-format PII with metric-DP transformations for semantic attributes. However, such approaches do not provide a unified mechanism for context-dependent entities such as names and locations or exact end-to-end restoration. Our approach similarly uses FPE for structured PII formats (e.g., SSNs, phone numbers, credit cards, etc.), while introducing {\em rank-class FPE} for dates to preserve valid calendar interpretations with fewer cryptographic operations.

\vspace{1mm}
\noindent{\bf \em User-Side and Zero-Trust Protection.}
User-side approaches keep sensitive data from reaching untrusted cloud services. Gim \etal~\cite{gim2024confidential} use trusted execution and multi-party decoding, while {\sc PII Shield}~\cite{holschneider2026pii} performs local LLM-based anonymization with decoy interactions. Chen \etal~\cite{chen2023hide} use local models to replace private entities and restore them in cloud responses. These approaches demonstrate the value of on-device protection but may introduce computational overhead or dependence on learned models. In contrast, our approach uses deterministic keyed mappings and FPE for lightweight, exact user-side restoration without trusted hardware or learned recovery.

\section{Threat Model}
\label{sec:threat}

We consider a user running \sys locally and interacting with a {\em remote} cloud LLM. Before transmission, \sys replaces detected PII with surrogates, so the LLM never receives the original values.

\noindent{\bf \em Attacker model.} We assume an \emph{honest-but-curious} remote LLM provider that performs inference correctly but may inspect, log, retain, analyze, and correlate inputs and outputs across sessions, possibly using auxiliary information. The provider knows the design and public parameters of \sys but cannot access user-side secret keys, surrogate mappings, or the restoration vault, compromise the user, or modify inference.

\noindent{\bf \em Security objectives.} We target three properties: (i) \emph{PII confidentiality}: detected original PII remains unrecoverable from observed surrogates; (ii) \emph{cross-session unlinkability}: the same underlying PII in different sessions is not directly linkable; and (iii) \emph{user-side reversibility}: the user can restore original PII from surrogates returned by the LLM.

\noindent{\bf \em Scope and assumptions.} Guarantees apply only to detected PII; undetected PII may remain exposed, and \sys does not prevent contextual or semantic inference. We assume an uncompromised user device and inaccessible user-side keys, mappings, and restoration state.

\begin{figure*}[t]
  \centering
  \includegraphics[width=0.9\textwidth]{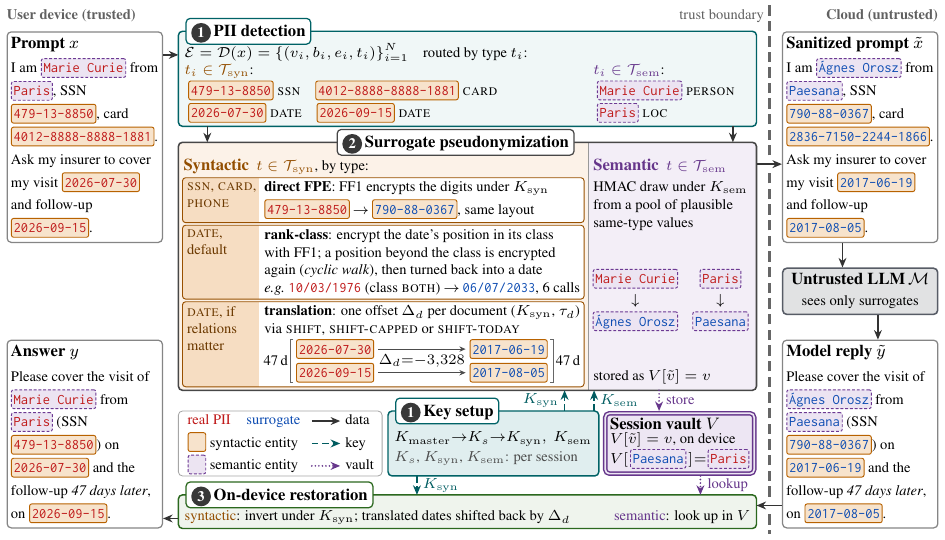}
    \caption{Overview of \sys's protection framework using an example prompt. Real PII is shown in red and surrogates in blue; orange solid frames denote syntactic handling and purple dashed frames semantic handling. (1)~A local detector $\mathcal{D}$ identifies and routes PII in prompt $x$, and \sys derives the required keys; (2)~Syntactic values are encrypted under $\Ksyn$ (or date-shifted by a document-level offset), while semantic values are replaced with same-kind surrogates under $K_{\mathrm{sem}}$ and stored in vault $V$ (Sections~\ref{sec:syntactic}--\ref{sec:semantic}). Only the sanitized prompt $\tilde{x}$, containing the generated surrogates, reaches the remote untrusted model $\mathcal{M}$; (3)~Syntactic surrogates in reply $\tilde{y}$ are inverted under $\Ksyn$, while semantic surrogates are restored from $V$, yielding the final response $y$ (Section~\ref{sec:restore}).}
  \label{fig:framework}
\end{figure*}

\section{Methodology} \label{sec:method}

We design \sys around a simple principle: {\em sensitive information should never leave the user's device in plaintext}. \sys locally detects PII, replaces it with realistic surrogates, and sends only the pseudonymized prompt to the untrusted cloud LLM. The response is deterministically restored on-device, completing a reversible round trip without exposing the original PII or its mappings. As shown in Fig.~\ref{fig:framework}, \sys comprises three components: {\em local PII detection, two-mechanism surrogate generation, and on-device restoration}.

\subsection{PII Detection and Protection} \label{sec:detect}

Given a user prompt \(x\), \sys identifies PII entities with a detector \(\mathcal{D}\) before the prompt is sent to the untrusted LLM \(\mathcal{M}\):
\begin{equation}
\mathcal{E} = \mathcal{D}(x)
= \{(v_i,b_i,e_i,t_i)\}_{i=1}^{N},
\label{eq:entities}
\end{equation}
where \(N\) is the number of detected PII entities, \(v_i = x[b_i:e_i]\) denotes the text span corresponding to the \(i\)-th PII entity, \(b_i\) and \(e_i\) specify its beginning and ending offsets in the prompt, respectively, and \(t_i \in \mathcal{T}\) denotes its PII type. Each detected entity is assigned to one of two protection methods based on its PII type: \emph{syntactic protection} or \emph{semantic protection} (see \sref{sec:surrogate}).

Before transformation, \sys establishes session-specific key material on the user's device. A master key \(\Km\) is generated locally and never leaves the device. For each session, \sys derives a session key \(\Ks\) from \(\Km\) using the session identifier \(\mathtt{sid}\) as session-specific context with the HMAC-based Key Derivation Function (HKDF) \cite{krawczyk2010hmac}:
\begin{equation}
\Ks = \mathrm{HKDF}(\Km, \mathtt{sid}).
\label{eq:sessionkey}
\end{equation}
Binding the derived key to the session identifier ensures that the same PII value can produce different protected representations across different sessions, thereby preventing straightforward cross-session linkability.
To further enforce key separation, \sys derives specific keys from \(\Ks\) using distinct labels for the {\em syntactic} and {\em semantic} mechanisms:
\begin{equation}
\begin{aligned}
K_{\mathrm{syn}} = \mathrm{HKDF}(\Ks, \text{``}\mathtt{syntactic}\text{''}),\\ \nonumber
K_{\mathrm{sem}} = \mathrm{HKDF}(\Ks, \text{``}\mathtt{semantic}\text{''}).
\end{aligned}
\label{eq:trackkeys}
\end{equation}
Thus, the two protection methods use independent session keys, {\em ensuring cross-session unlinkability.}


\subsection{PII-Aware Surrogate Pseudonymization} \label{sec:surrogate}

After PII detection and type-based assignment, \sys replaces each detected entity with a session-specific surrogate while preserving information needed for downstream model reasoning. It uses two complementary protection mechanisms: a \emph{syntactic mechanism} for structured-format PII, such as $\Tstr=\{\textsc{ssn}, \textsc{phone}, \textsc{card}, \textsc{date}\}$, and a \emph{semantic mechanism} for PII that requires preserving its type and linguistic plausibility, such as $\Tsem=\{\textsc{person}, \textsc{loc}, \textsc{org}, \textsc{email}\}$.

\subsubsection{Syntactic  PII protection Method} \label{sec:syntactic}

For $t\in\Tstr$, a surrogate must (a) remain valid for its type, (b) preserve
the original format layout, and (c) be reversible using only the syntactic key $\Ksyn$ (Section~\ref{sec:restore}). \sys provides two keyed mechanisms:
\textbf{\em Encipherment}, which independently applies FPE to each PII; and \textbf{\em Translation}, used for dates when temporal relations must be preserved.

\vspace{1mm}
\noindent\underline{\bf\em 1. Encipherment Mechanism:}
We use FF1 standardized in NIST SP~800-38G \cite{dworkin2001recommendation}, denoting encryption under $\Ksyn$ with type-specific tweak $\tau_t$ by $E_{\Ksyn,\tau_t}$ for domain separation. We use two regimes depending on whether every digit string of the required length is valid; dates require dedicated handling. For \underline{\em direct} identifiers such as SSNs, phone numbers, and card numbers, let $\delta(v)$ extract the digits and $\pi(v)$ preserve the non-digit layout. The surrogate is
\begin{equation}
\tilde{v}=\pi(v)\oplus E_{\Ksyn,\tau_t}(\delta(v)),
\label{eq:template}
\end{equation}
where $\oplus$ inserts the enciphered digits into their original positions, preserving punctuation, spacing, and grouping while transforming the underlying value. Note that a leading telephone country code is copied unchanged. Appendix~\ref{app:worked} provides examples for each syntactic PII type, with Table~\ref{tab:worked} detailing the encipherment steps.

For \underline{\em dates}, validity depends on the {\em calendar semantics} rather than the {\em digit sequence} alone. For a four-digit year, $10^8$ digit strings are possible, but only $73{,}049$ represent valid dates within \sys's supported range (1900--2099), so direct digit encipherment can produce {\em invalid} dates. We therefore use rank-then-encipher \cite{bellare2010ffx}. Let $\mathcal{V}(v)$ be the valid values compatible with the representation of $v$, and let
\begin{equation}
\rho:\mathcal{V}(v)\rightarrow\{0,\ldots,n-1\},\qquad n=|\mathcal{V}(v)|
\end{equation}
be a bijective ranking. The surrogate is
\begin{equation}
\tilde{v}=
\rho^{-1}\!\left(E_{\Ksyn,\tau_t}(\rho(v))\right),
\label{eq:rank}
\end{equation}
which is valid by construction. FF1 operates at the smallest decimal width $w$ satisfying $10^w\ge n$; outputs $\ge n$ are re-enciphered until they fall within the valid domain (\emph{cyclic walk}). The valid set is derived from the input Value and never stored.


\noindent \sys supports five date representations:

\begin{enumerate}[leftmargin=*,itemsep=1pt,topsep=1pt,parsep=0pt]
\item \textbf{\em Fixed-layout dates:} For dates such as \texttt{1985-07-24}, \sys parses and ranks the date within its valid calendar domain, then renders the surrogate in the original layout, preserving field order, widths, and separators.

\item \textbf{\em Ambiguous numeric dates:} For dates such as \texttt{10/03/1976} that admit multiple interpretations, \sys ranks them within their ambiguity class
\begin{equation}
\begin{aligned}
C(v) &\in
\{\textsc{both},\textsc{mdy},\textsc{dmy},\textsc{ymd}\},\\\nonumber
\mathcal{V}(v) &= \{u : C(u)=C(v)\},
\end{aligned}
\label{eq:class}
\end{equation}
thereby preserving the interpretations supported by the input without assuming a locale.

\item \textbf{\em Textual dates:} For dates such as \texttt{Feb 4th, 1909}, \sys canonicalizes the date to an unambiguous format, ranks, and enciphers it, then renders the surrogate in the original textual style.

\item \textbf{\em Timestamps:} For timestamps such as \texttt{1902-08-21T22:37:28.285Z}, \sys transforms the date and time-of-day components using separate tweaks while preserving separators, fractional precision, and time zone.

\item\textbf{\em Partial dates:} For incomplete dates such as \texttt{20th September}, \sys declines the entity because it lacks a complete calendar date suitable for rank-based encipherment.
\end{enumerate}

 Detailed examples are provided in Appendix~\ref{app:worked}.

\vspace{1mm}
\noindent\underline{\bf\em 2. Translation Mechanism:} 
While independent date encipherment preserves  privacy and enables stateless restoration, it breaks temporal relations such as \emph{order}, \emph{duration}, and \emph{tense}. Since these relations depend on date pairs, \sys instead provides document-level translation as an alternative.

Let $N$ be the number of days in the specified range, $\nu$ map dates bijectively to $\mathbb{Z}_N=\{0,\ldots,N-1\}$, and
$a_i=\nu(v_i)$ for dates $\{v_1,..,v_k\}$ in document $d$. Let $t_0=\nu(\mathrm{today})$, $\tau_d$ be a document-specific tweak derived from the document identifier, and $F_{\Ksyn}(x)$ be an HMAC-based pseudorandom function under $\Ksyn$. Given a system-wide displacement bound $S$, \sys deterministically selects a single admissible offset
$\Delta_d\in[-S,S]$:
\begin{equation}
\begin{aligned}
\delta_j &=
-S+\big(F_{\Ksyn}(\tau_d\Vert j)\bmod(2S+1)\big),\\
\Delta_d&=\delta_{j^\ast},\qquad
j^\ast=\min\{\,j\ge0:\delta_j\in\mathcal{A}^{\ast}\,\},
\end{aligned}
\label{eq:Off_selec}
\end{equation}
where $\mathcal{A}^{\ast}$ excludes recognizable shifts, such as those preserving weekdays or seasons. Since $\mathcal{A}^{\ast}$ and the candidate sequence are plaintext- independent, $\Delta_d$ depends only on $(\Ksyn,d)$ and need not be stored. The same offset is applied to every date:
\begin{equation}
\tilde{v}_i=
\nu^{-1}\!\left((a_i+\Delta_d)\bmod N\right),
\label{eq:same_off}
\end{equation}
with the original layout retained during rendering. Translation is a bijection on $\mathbb{Z}_N$, preserving validity, reversibility, and format. Applying the same offset to all dates preserves their relative order and duration unless a date crosses the supported-range boundary (Appendix~\ref{app:translation}). Thus, encipherment independently permutes dates, whereas translation applies a common keyed rotation. Details and acceptance rates of offsets are given in Appendix~\ref{app:tr-params}.

Importantly, translation alone does not preserve \emph{tense}, since the present $t_0$
remains fixed. \sys therefore supports two variants:
\begin{itemize}[leftmargin=*]
\item \textbf{\textsc{shift-capped}:} restricts the offset to values that keep
all dates on their original side of $t_0$:
\begin{equation}
\resizebox{6.4cm}{!}{$
\begin{aligned}
W_d = \big[\, &t_0 - \min\nolimits_{a_i \ge t_0} a_i + 1,\\
              &t_0 - \max\nolimits_{a_i < t_0} a_i - 1 \,\big] \cap [-S, S],
\end{aligned}$}
\label{eq:capped}
\end{equation}
and samples from $W_d\cap\mathcal{A}^{\ast}$. This preserves tense but makes the shift plaintext-dependent and requires storing the offset.


\item \textbf{\textsc{shift-today}:} keeps Eqs.~\eqref{eq:Off_selec}--\eqref{eq:same_off}
unchanged and shifts the stated present by the same offset,
$\tilde{t}_0=t_0+\Delta_d$. Hence \(a_i<t_0 \iff a_i+\Delta_d<\tilde{t}_0,\) preserving tense without limiting the displacement or storing $\Delta_d$.  This assumes that $\mathcal{M}$ evaluates temporal relations relative to the prompt-stated present \cite{gao2025can}.
\end{itemize}
Thus, translation trades independent encipherment for document-level temporal consistency and is used when such relations matter. Table~\ref{tab:temporal} in Appendix~\ref{app:tr-example} compares \textsc{shift}, \textsc{shift-capped}, and \textsc{shift-today} in preserving {\em order, duration, tense}.

 \subsubsection{Semantic PII Protection Mechanism} \label{sec:semantic}

For $t\in\Tsem$, character-level encipherment can destroy an entity's linguistic identity. \sys rather replaces each value with a plausible surrogate of the same {\em semantic type}, preserving its class and relevant attributes. Each value $v$ is normalized to $\bar v$ by collapsing whitespace and lower-casing, mapping formatting variants to exact canonical value.

\noindent\underline{\bf \em Pools and Selection.} For each semantic bucket $b$, \sys maintains a curated pool $P_b$ of type-consistent values. Pools are constructed from Faker-Library -generated locale data over fifteen Latin-script locales \cite{faraglia_faker}, with deterministic ordering (Appendix~\ref{app:pools}). Persons use locale- and gender-consistent given-name and surname pools to maintain plausible name combinations; locations use city names, organizations use company names, and e-mails combine independently selected names and domains. Given the canonical value $\bar v$, \sys deterministically selects a surrogate using HMAC under the semantic session key $K_{\mathrm{sem}}$:
\begin{equation}
\resizebox{6.9cm}{!}{$
\begin{aligned}
\iota &=
\mathrm{HMAC}\big(K_{\mathrm{sem}},,b\Vert r\Vert\bar v\big)
\bmod |P_b|,\
\tilde v &= P_b[\iota],
\end{aligned}
$}
\label{eq:dict}
\end{equation}
where $r$ is a candidate index. Attribute-bearing buckets use independent
instances of the same construction under attribute-specific labels. Because
$K_{\mathrm{sem}}$ is session-specific, the same value maps to different
surrogates across sessions, while determinism preserves consistency within a
session.

\noindent\underline{\bf \em Collision resolution.} Distinct values must remain distinguishable. If the selected $P_b[\iota]$ has already been assigned to another value, \sys increments $r$ and re-evaluates Eq.~\eqref{eq:dict} until an unused candidate is obtained. For the $k$-th distinct value, the first candidate is rejected with probability at most $(k-1)/|P_b|$, which is small for \sys's large pools.


\noindent\underline{\bf \em Vault.}
Unlike the syntactic mechanism, semantic substitution is not reversible from
the key alone because HMAC is one-way. \sys therefore maintains a session-local
vault,
\begin{equation}
V[\tilde v]=v,
\label{eq:semantic-map}
\end{equation}
together with a reverse index from $\bar v$ to $\tilde v$ for repeated values
and collision tracking. The vault remains on-device and is erased when the session ends; restoration is described in Section~\ref{sec:restore}.

Thus, the semantic mechanism trades stateless syntactic encipherment for plausible, type-consistent surrogates and session-local state. Table~\ref{tab:semantic} in Appendix~\ref{app:pools} provides a worked example of canonicalization, keyed surrogate selection, collision handling, and cross-session variation.


 \subsection{On-Device PII  Restoration}
\label{sec:restore}

After surrogate generation, \sys replaces each detected PII entity in the
prompt with its surrogate, producing sanitized prompt $\tilde{x}$ sent to the
untrusted model $\mathcal{M}$. The returned response $\tilde{y}$ is restored
locally before being shown to the user. For each surrogate $\tilde{v}$, \sys
applies the inverse of its generation mechanism:
\begin{equation}
v =
\begin{cases}
\mathcal{R}_{\Ksyn,t}(\tilde{v}),
    & t \in \Tstr, \\[4pt]
V[\tilde{v}],
    & t \in \Tsem,
\end{cases}
\label{eq:restore}
\end{equation}
where $\mathcal{R}_{\Ksyn,t}$ reverses the corresponding syntactic
transformation and $V$ is the session-local semantic vault. For directly
enciphered values, restoration decrypts the digit string and refills the
original template; for ranked domains, it decrypts the rank and un-ranks it
within the recomputed valid domain. For translated dates, $\Delta_d$ is
re-derived from $\Ksyn$ and $\tau_d$ and removed from the cyclic day line:
\begin{equation}
v_i =
\nu^{-1}\!\left((\nu(\tilde{v}_i)-\Delta_d)\bmod N\right).
\label{eq:untranslate}
\end{equation}
Since translation is a bijective rotation, it preserves model-computed dates relative to their surrogates.

All restoration state remains on-device: $\Km$, $\Ks$, and the semantic vault $V$ are never exposed to $\mathcal{M}$ and are discarded at session termination. Thus, the remote LLM sees only surrogate-bearing text, while \sys restores the response locally.

\vspace{-1.5mm}

\section{Performance Evaluation}
\label{sec:eval}

\vspace{-1.5mm}
\subsection{Experimental Setup}
\label{sec:setup}

\noindent{\bf\em Datasets.} We evaluate \sys on four corpora: ai4privacy \cite{ai4privacy_pii_masking_200k}, PUPA \cite{siyan2025papillon}, WildChat \cite{zhao2024wildchat}, and ShareGPT \cite{sharegpt2023vicuna}; details appear in Appendix~\ref{app:datasets}.

\vspace{0.5mm}
\noindent{\bf\em Models.} We evaluate three untrusted models from different vendors: \texttt{Llama-3.2-3B} \cite{grattafiori2024llama3herdmodels}, \texttt{Gemma-2-2B} \cite{team2024gemma}, and \texttt{Phi-3.5-mini} \cite{abdin2024phi3technicalreporthighly}, each evaluated under identical environment and cryptographic parameters (Appendix~\ref{app:compute}).

\vspace{0.5mm}
\noindent{\bf\em Baselines.} We compare \sys with {\em eleven} state-of-the-art protection schemes, grouped by their trust boundary (Appendix~\ref{app:baselines}). \emph{Entity substitution} includes Pr$\varepsilon\varepsilon$mpt \cite{chowdhury2025pr}, Rescriber \cite{zhou2025rescriber}, and Hide-and-Seek (HaS) \cite{chen2023hide}; \emph{local rewrite} includes PAPILLON \cite{siyan2025papillon}, Staab \etal's adversarial anonymizer \cite{staab2402large}, and SEAL \cite{kim2026self}; and \emph{DP} includes SanText+ \cite{yue2021differential}, CusText \cite{chen2023customized}, DP-MLM \cite{meisenbacher2024dp}, DP-Prompt \cite{utpala2023locally}, and InferDPT \cite{tong2025inferdpt}.

\vspace{0.5mm}
\noindent{\bf\em Evaluation metrics.}
We evaluate 4 aspects of \sys: {\em detection, transformation, privacy at the trust boundary, and end-to-end restoration}.
For detection, we report {\em Recall} and {\em Missed}. For transformation, we measure {\em Protected, Round trip, Valid,} and {\em Collisions}, while date-specific evaluation includes {\em Ambiguity kept, Cipher calls,} and translation consistency.
To assess what crosses the trust boundary, we measure {\em Leaked, Leak-free, Retained, Perplexity (PPL), Unchanged, Words changed,} and {\em Local generations}. Finally, end-to-end performance is evaluated using {\em Echoed, Restored, Untrusted answer, Local answer, Value back,} and {\em Restored by}. Metric definitions are provided in Appendix~\ref{app:metrics}.

\begin{table*}[!t]
\centering
\small
\setlength{\tabcolsep}{2pt}
\renewcommand{\arraystretch}{1}
\resizebox{13cm}{!}{
\begin{tabular}{@{}l|l|cc|cc|cc|c|l@{}}
\toprule
& & \multicolumn{2}{c|}{Privacy}
& \multicolumn{2}{c|}{Utility}
& \multicolumn{2}{c|}{Prompt Edit}
& Cost & Restoration \\
\cmidrule(lr){3-4}
\cmidrule(lr){5-6}
\cmidrule(lr){7-8}
\cmidrule(lr){9-9}
\cmidrule(lr){10-10}
\raisebox{ 2pt}[1pt][2pt]{\rotatebox[origin=c]{90}{\text{Category}}}
 & Method
& Leaked$\downarrow$
& Leak-free$\uparrow$
& Retained$\uparrow$
& PPL ($\times$)
& Unchanged$\uparrow$
& Changed$\downarrow$
& Local queries$\downarrow$
& Restored by \\
& 
& (PII values)
& (Prompts)
& (Words)
&
& (Prompts)
& (Words)
& (Number)
& \\
\midrule

& Unprotected
& 95.1 & 3.8 & 100 & 1.00 &   100& 0.0 & 0 & -- \\
\midrule

\multirow{4}{*}{\rotatebox[origin=c]{90}{\shortstack{Entity\\Substitute}}}
& Rescriber
& 27.0 & 64.8 & 93 & 1.00 & 11.7 & 10.6 & 0   & map \\
& HaS
& 19.8 & 69.3 & 86 & 1.02 & 12.5 & 18.4 & 1.9 & model \\
& Pr$\varepsilon\varepsilon$mpt
& 63.8 & 25.5 & 99 & 1.00 & 52.3$^\dagger$ & 3.1  & 0   & key, map \\
& \textbf{\sys}
& {27.0} &  {64.6} &  {95} &  {1.20}
&  {11.9} &  {9.4} & \textbf{0} & \textbf{key, vault} \\
\midrule

\multirow{3}{*}{\rotatebox[origin=c]{90}{\shortstack{Local \\rewrite}}}

& PAPILLON
& 21.7 & 60.5 & 38 & 0.47 & 0.0 & 78.6 & 2.0 & model \\
& Staab \etal
& 6.5 & 90.5 & 62 & 0.81 & 0.3 & 43.9 & 6.0 & not restored \\
& SEAL
& \textbf{0.2} & \textbf{99.4} & 13 & 0.76 & 0.0 & 86.0 & 6.0 & not restored \\
\midrule

\multirow{5}{*}{\rotatebox[origin=c]{90}{DP-based}}
& SanText+ $\varepsilon\!=\!3$
& 51.5 & 37.7 & 84 & 8.53 & 0.6 & 38.1 & 0   & not restored \\
& CusText $\varepsilon\!=\!3$
& 55.0 & 32.1 & 45 & 38.9 & 2.3 & 71.7 & 0   & not restored \\
& DP-MLM $\varepsilon\!=\!25$
& 4.4 & 91.1 & 15 & 30.2 & 0.0 & 64.1 & 0   & not restored \\
& DP-Prompt $\varepsilon\!=\!53.4$
& 23.2 & 61.1 & 27 & 3.21 & 1.1 & 79.5 & 0   & not restored \\
& InferDPT $\varepsilon\!=\!3$
& 0.4 & 98.9 & 5 & 344.4 & 0.0 & 98.5 & 1.0 & model \\
\bottomrule
\end{tabular}}
\caption{Trust-boundary leakage, utility, prompt edits, and restoration cost on $664$ PUPA rows using \texttt{Llama-3.2-3B}. Arrows indicate the direction; values are in \% except PPL and cost, reported as a ratio and count, respectively. $^\dagger$Pr$\varepsilon\varepsilon$mpt has no mechanism for locations, organizations or e-mail addresses: it sends $52.3\%$ of prompts unchanged, and $329$ of these $347$ still expose PII verbatim, whereas \sys, given the same detected entities, edits $268$ of them.}
\label{tab:crossed}
\end{table*}
 \begin{table*}[!t]
\centering\small
\setlength{\tabcolsep}{2pt}
\renewcommand{\arraystretch}{1}
\resizebox{12.3cm}{!}{
\begin{tabular}{@{}l|l|cccc|ccccc|ccc@{}}
\toprule
& & \multicolumn{4}{c|}{\bf Syntactic types}
& \multicolumn{5}{c|}{\bf Semantic types}
& \multicolumn{3}{c}{\bf Aggregate} \\
\cmidrule(lr){3-6}
\cmidrule(lr){7-11}
\cmidrule(lr){12-14}
\raisebox{7pt}[0pt][0pt]{\rotatebox[origin=c]{90}{{Metric}}}
 
& \textbf{Method}
& \textsc{ssn} & \textsc{phone} & \textsc{card} & \textsc{date}
& \textsc{person} & \textsc{loc} & \textsc{org} & \textsc{email}
& \shortstack{\textsc{date}}
& Syntactic & Semantic & All \\

\midrule
& {Detected entities} (\#)
& 475 & 644 & 1{,}014 & 1{,}147
& 4{,}589 & 2{,}642 & 2{,}602 & 887 & 2{,}405
& 3{,}280 & 13{,}125 & 16{,}405 \\
\midrule
\multirow{4}{*}{\rotatebox[origin=c]{90}{\textit{Protected}}}

& Pr$\varepsilon\varepsilon$mpt
& 100 & 99.2 & 100 & 93.2 & 100 & n/a & n/a & n/a & 0.1
& 97.5 & 35.0 & 47.5 \\
& Rescriber
& 100 & 100 & 100 & 100 & 100 & 100 & 100 & 100 & 100
& 100 & 100 & 100 \\
& \sys
& 100 & 99.4 & 100 & 100 & 100 & 100 & 100 & 100 & 67.2
& 99.9 & 94.0 & 95.2 \\
& \sys--w/o FPE
& 100 & 100 & 100 & 100 & 99.7 & 100 & 99.0 & 100 & 67.0
& 100 & 93.6 & 94.9 \\
\midrule

\multirow{4}{*}{\rotatebox[origin=c]{90}{\textit{Round trip}}}

& Pr$\varepsilon\varepsilon$mpt
& 100 & 100 & 100 & 100 & 72.9 & n/a & n/a & n/a & 100
& 100 & 72.9 & 84.0 \\
& Rescriber
& 100 & 100 & 100 & 100 & 100 & 100 & 100 & 100 & 100
& 100 & 100 & 100 \\
& \sys
& 100 & 100 & 100 & 100 & 100 & 100 & 99.2 & 100 & 99.8
& 100 & 99.8 & 99.9 \\
& \sys--w/o FPE
& 100 & 100 & 100 & 100 & 100 & 100 & 99.2 & 100 & 99.8
& 100 & 99.8 & 99.9 \\
\midrule

\multirow{4}{*}{\rotatebox[origin=c]{90}{\textit{Valid}}}

& Pr$\varepsilon\varepsilon$mpt
& 48.4 & 13.1 & 7.2 & 0.3 & 98.8 & n/a & n/a & n/a & 0.0
& 12.2 & 98.8 & 63.2 \\
& Rescriber
& 0.0 & 0.0 & 9.8 & 0.0 & 0.0 & 0.0 & 0.0 & 0.0 & 0.0
& 3.0 & 0.0 & 0.6 \\
& \sys
& 48.0 & 13.3 & 8.5 & 88.4 & 100 & 100 & 100 & 92.3 & 0.1
& 43.1 & 86.4 & 77.3 \\
& \sys--w/o FPE
& 46.9 & 16.8 & 9.3 & 65.9 & 100 & 100 & 100 & 88.2 & 0.1
& 36.0 & 86.0 & 75.5 \\
\midrule

\multirow{4}{*}{\rotatebox[origin=c]{90}{\textit{Collisions}}}

& Pr$\varepsilon\varepsilon$mpt
& 0 & 0 & 0 & 0 & 875 & n/a & n/a & n/a & 0 & 0 & 875 & 875 \\
& Rescriber
& 0 & 0 & 0 & 0 & 0 & 0 & 0 & 0 & 0 & 0 & 0 & 0 \\
& \sys
& 0 & 0 & 0 & 0 & 0 & 0 & 0 & 0 & 0 & 0 & 0 & 0 \\
& \sys--w/o FPE
& 0 & 0 & 0 & 0 & 0 & 0 & 0 & 0 & 0 & 0 & 0 & 0 \\
\bottomrule
\end{tabular}}
\caption{Transformation performance across personal-data types on $10{,}000$ ai4privacy documents.}
\label{tab:pertype}
\end{table*}

\subsection{Main Results: \sys vs.~Baselines}

\noindent{\bf\em  Privacy--utility trade-off.} 
On all $664$ PUPA rows (Table~\ref{tab:crossed}), \sys reduces Leaked from $95.1\%$ to $27.0\%$ ($3.5\times$) while retaining $95\%$ of the request {\em without requiring a user-side LLM query to restore the generated response}. Among entity-substitution methods, \sys matches Rescriber at $27.0\%$ leakage and improves over Pr$\varepsilon\varepsilon$mpt, which leaks $2.4\times$ more while leaving several contextual PII types unprotected. HaS lowers leakage further to $19.8\%$, but retains only $86\%$ of the request and requires about two Local generations per request. Other methods achieve lower leakage by retaining substantially less context: $62\%$ (Staab \etal), $38\%$ (PAPILLON),$27\%$ (DP-Prompt), $15\%$ (DP-MLM), $13\%$ (SEAL), and $5\%$ (InferDPT). At their published $\varepsilon$, SanText+ and CusText also leak $1.9\times$ and $2.0\times$ more than \sys. Thus, \sys offers a strong privacy--utility trade-off, combining $27.0\%$ leakage with $95\%$ retention and no local generation, while preserving realistic, same-type surrogates.

\begin{table*}[t]
\centering\small
\setlength{\tabcolsep}{2pt}
\renewcommand{\arraystretch}{0.9}
\resizebox{13.6cm}{!}{
\begin{tabular}{lcccccccccccl}
\toprule
 & \multicolumn{9}{c}{Leaked (\% values)} & & & \\
\cmidrule(lr){2-10}
Method & \textsc{person} & \textsc{loc} & \textsc{org} & \textsc{email} & \textsc{ssn} & \textsc{phone} & \textsc{card} & \textsc{date} & All & \shortstack{Retained\\(\% words)} & \shortstack{Value\\back (\%)} & \shortstack{Restored\\by} \\
\emph{gold entities} (\#) & 462 & 183 & 65 & 84 & 33 & 67 & 63 & 149 & 1{,}106 & & & \\
\midrule
Unprotected &   100&   100&   100&   100&   100&   100&   100&   100&   100&   100&   38.4 & -- \\
\midrule

Pr$\varepsilon\varepsilon$mpt &  8.2 &   98.9 &   75.4 &   100&  3.0 &  9.0 &  1.6 &   33.6 &   37.1 &   99.9 &   39.5 & key, map \\
Rescriber &  3.0 &  4.9 &  1.5 &  0.0 &  3.0 &  9.0 &  0.0 &  0.0 &  2.8 &   98.0 &   47.0 & map \\
HaS &  3.2 &  3.8 &  1.5 &  0.0 &   24.2 &   13.4 &   23.8 &  2.0 &  5.2 &   96.5 &   42.6 & model \\
{\bf \sys} &  3.2 &  4.9 &  1.5 &  0.0 &  3.0 &  9.0 &  0.0 &  0.0 &  2.9 &   99.1 &   38.2 & key, vault \\
\midrule
PAPILLON &  4.5 &   43.7 &   20.0 &  3.6 &  0.0 &  7.5 &  4.8 &   12.1 &   12.9 &   52.2 &   26.0 & model \\
Staab \etal &  2.6 &  0.0 &  3.1 &  0.0 &  3.0 &  0.0 &  0.0 &  4.0 &  1.9 &   76.7 &  1.0 & not restored \\
SEAL &  {\bf 0.2}&  0.0 &  0.0 &  0.0 &  0.0 &  0.0 &  0.0 &  0.0 &  0.1 &   15.3 &  0.0 & not restored \\
\midrule
SanText+ $\varepsilon\!=\!3$ &   78.6 &   62.3 &   21.5 &  0.0 &  3.0 &  0.0 &  0.0 &   15.4 &   46.6 &   89.1 &   22.1 & not restored \\
CusText $\varepsilon\!=\!3$ &   80.1 &   49.2 &   10.8 &   100&   81.8 &   41.8 &   15.9 &   70.5 &   65.2 &   31.8 &   21.0 & not restored \\
DP-MLM $\varepsilon\!=\!25$ &  3.9 &  2.2 &  0.0 &  0.0 &  0.0 &  0.0 &  0.0 &  0.0 &  2.0 &  9.1 &  1.1 & not restored \\
DP-Prompt $\varepsilon\!=\!53.4$ &   49.8 &   48.6 &   40.0 &   32.1 &   42.4 &   29.9 &   36.5 &   42.3 &   44.5 &   44.5 &   20.6 & not restored \\
InferDPT $\varepsilon\!=\!3$ &  0.4 &  0.5 &  0.0 &  0.0 &  0.0 &  0.0 &  0.0 &  0.0 &  0.3 &  1.8 &   57.9 & model \\
\bottomrule
\end{tabular}}
\caption{Leaked, per personal-data type, \sys and every baseline on the same
$1{,}000$ ai4privacy documents, \texttt{Llama-3.2-3B} as the untrusted model.}
\label{tab:survives}
\end{table*}


\noindent{\bf\em PII Leakage by Data Type.} Table~\ref{tab:survives} extends Table~\ref{tab:crossed} with per-type results on $1{,}000$ ai4privacy documents. \sys leaks $2.9\%$ of values while retaining $99.1\%$ of the request; Rescriber, operating on the same entities, is comparable at $2.8\%$ leakage and $98.0\%$ retention. Six baselines leak more: HaS $1.8\times$, with its hider leaving $24.2\%$ of \textsc{ssn} and $23.8\%$ of \textsc{card} values; PAPILLON $4.4\times$; Pr$\varepsilon\varepsilon$mpt $13\times$, with no mechanism for three types; and SanText+, DP-Prompt, and CusText $15$--$22\times$. The four methods with lower leakage substantially reduce Retained: Staab \etal~to $76.7\%$, SEAL to $15.3\%$, DP-MLM to $9.1\%$, and InferDPT to $1.8\%$.

\begin{table}[!t]
\centering\small
\setlength{\tabcolsep}{2pt}
\renewcommand{\arraystretch}{0.9}
\resizebox{\linewidth}{!}{
\begin{tabular}{l r cc cc cc}
\toprule
& & \multicolumn{2}{c}{\texttt{Llama-3.2-3B}}
& \multicolumn{2}{c}{\texttt{Gemma-2-2B}}
& \multicolumn{2}{c}{\texttt{Phi-3.5-mini}} \\
\cmidrule(lr){3-4}\cmidrule(lr){5-6}\cmidrule(lr){7-8}
\textbf{Method} & $n$ 
& Echoed & Restored
& Echoed & Restored
& Echoed & Restored \\
& (\#)
& \% (\#) & (\%)
& \% (\#) & (\%)
& \% (\#) & (\%) \\
\midrule
\sys
& 863 & 42.2 (364) & 100 & 48.7 (420) & 100 & 79.7 (688) & 100 \\
\sys--w/o FPE
& 862 & 38.5 (332) & 100 & 46.5 (401) & 100 & 75.5 (651) & 100 \\
Pr$\varepsilon\varepsilon$mpt
& 442 & 34.6 (153) & 100 & 42.1 (186) & 100 & 79.4 (351) & 100 \\
Rescriber
& 912 & 48.2 (440) & 100 & 43.0 (392) & 100 & 69.3 (632) & 100 \\
\bottomrule
\end{tabular}}
\caption{End-to-end restoration on ai4privacy using $500$ prompts and three remote unttrusted LLMs.}

\label{tab:e2e_stored}
\end{table}
\noindent{\bf\em Prompt perplexity.} Under \texttt{Llama-3.2-3B}, \sys sends text at $1.20\times$ the unprotected prompt's perplexity, slightly above entity-substitution methods ($1.00$--$1.02$) but far below DP methods: $2.7\times$ (DP-Prompt), $7.1\times$ (SanText+), $25\times$ (DP-MLM), $32\times$ (CusText), and $287\times$ (InferDPT). Pr$\varepsilon\varepsilon$mpt and Rescriber remain at $1.00$ because the former leaves $52.3\%$ of prompts unchanged, while the latter's predictable placeholders reduce perplexity. In contrast, \sys changes fewer words than Rescriber ($9.4\%$ vs.~$10.6\%$) but raises perplexity because its surrogates are key-selected rather than context-predicted. PAPILLON ($0.47$), SEAL ($0.76$), and Staab \etal~($0.81$) fall below $1.00$ by rewriting and discarding $62\%$, $87\%$, and $38\%$ of the prompt, respectively. Thus, perplexity deviations reflect rewriting or obfuscation; \sys remains near the original while retaining $95\%$ of the prompt.

\noindent{\bf\em  Surrogate prompt quality.}
Table~\ref{tab:pertype} compares \sys with \sys w/o FPE, Pr$\varepsilon\varepsilon$mpt, and Rescriber on the same detected entities. Dates are split into \textsc{date} ($1{,}147$ ciphered dates) and \textsc{date fallback} ($2{,}405$ dates, times, ages, and relative expressions), with digit-containing values vaulted. \sys protects at least $99.4\%$ of every type except \textsc{date fallback} ($67.2\%$), similar to \sys w/o FPE ($67.0\%$). Pr$\varepsilon\varepsilon$mpt protects $47.5\%$ overall, lacking mechanisms for locations, organizations, and emails. \sys round-trips $99.9\%$ of protected entities without collisions, versus $72.9\%$ of names for Pr$\varepsilon\varepsilon$mpt, which reuses $875$ surrogates across names. Valid surrogates reach $77.3\%$ for \sys, versus $75.5\%$, $63.2\%$, and $0.6\%$ for \sys w/o FPE, Pr$\varepsilon\varepsilon$mpt, and Rescriber. FPE drives this gain: \textsc{date} validity is $88.4\%$, versus $65.9\%$ without FPE and $0.3\%$ for Pr$\varepsilon\varepsilon$mpt.
Thus, \sys  achieves high type validity, $99.9\%$ round-trip restoration, and zero collisions.

\noindent{\bf\em Restoration across models.} Table~\ref{tab:e2e_stored} traces each value through three untrusted models. Stored-state methods restore all echoed surrogates correctly ($100.0\%$), making recovery dependent on Echoed: \sys recovers $42.2\%$, $48.7\%$, and $79.7\%$ on Llama, Gemma, and Phi, respectively. Its Echoed rate varies $1.9\times$ across models, versus at most $1.4\times$ for competitors, indicating that recovery is primarily limited by the model's tendency to repeat values. Local-regeneration results appear in Appendix~\ref{app:e2e_local}, with additional value-recovery results in Appendix~\ref{app:permethod}, Tables~\ref{tab:survives} and~\ref{tab:valueback}.

\begin{table}[!t]
\centering
\footnotesize
\setlength{\tabcolsep}{2pt}
\resizebox{7cm}{!}{
\begin{tabular}{c|lcc}
\toprule
& Scheme &  {Ambiguity kept (\%)} & {Cipher calls (\#)}    \\
\midrule
\multirow{3}{*}{%
  \rotatebox[origin=c]{90}{%
    \scriptsize\shortstack{\emph{ai4privacy}\\[-2pt]
    $628$ dates}%
  }}
& walk & 37.8 & 1{,}072 \\
& rank & 54.8 & 1.7 \\
& rank-class (\sys) & \textbf{100} & 2.6 \\
\midrule
\multirow{3}{*}{%
  \rotatebox[origin=c]{90}{%
    \scriptsize\shortstack{\emph{WildChat}\\[-2pt]
    $328$ dates}%
  }}
& walk & 96.9 & 1{,}022 \\
& rank & 95.5 & 7.7 \\
& rank-class (\sys) & \textbf{100} & 1.4 \\
\midrule
\multirow{3}{*}{%
  \rotatebox[origin=c]{90}{%
    \scriptsize\shortstack{\emph{ShareGPT}\\[-2pt]
    $879$ dates}%
  }}
& walk & 86.8 & 1{,}060 \\
& rank & 92.8 & 42.1 \\
& rank-class (\sys) & \textbf{100} & 1.8 \\
\bottomrule
\end{tabular}}
\caption{Date encipherment schemes under FF1. The corpora contain $628$, $328$, and $879$ dates, respectively.}
\label{tab:ambiguity} 
\end{table}
 \begin{table}[!t]
\centering
\footnotesize
\setlength{\tabcolsep}{2pt}
\resizebox{7.5cm}{!}{
\begin{tabular}{l|ccc|cc}
\toprule
& \multicolumn{3}{c|}{Temporal consistency (\%)}
& \multicolumn{2}{c}{Date concealment} \\
\cmidrule(lr){2-4}\cmidrule(lr){5-6}
Scheme & Order & Tense & \shortstack{Computed date}
& \shortstack{Shifted (days)} & \shortstack{Unshifted (\#)} \\
\midrule
Unprotected & 100 & 73.2 & 100 & 0 & 97 \\
Encipherment & 53.6 & 24.7 & 0.0 & 11{,}242 & 0 \\
\midrule
\multicolumn{6}{l}{\textbf{\em Translation}} \\
\midrule
\textsc{shift} & 100 & 5.2 & 97.9 & 2{,}224 & 0 \\
\textsc{shift-capped} & 97.9 & 75.3 & 94.8 & 34 & 11 \\
\textsc{shift-today} & 100 & 77.3 & 97.9 & 2{,}224 & 0 \\
\bottomrule
\end{tabular}}
\caption{End-to-end date translation on $97$ ShareGPT files (two dates) using \texttt{Llama-3.2-3B} (remote LLM).}
 
\label{tab:translation}
\end{table}

\subsection{Ablation Study}\label{sec:ablation}

\noindent{\bf\em Choice of Date scheme.} Table~\ref{tab:ambiguity} compares \sys's rank-class scheme with walk and rank. Rank-class preserves $100.0\%$ of ambiguity across all corpora, vs. $37.8$--$96.9\%$ for walk and $54.8$--$95.5\%$ for rank. It also requires $1.4$--$2.6$ cipher calls per date, $410$--$730\times$ fewer than walk; rank require $23\times$ more calls on ShareGPT. Thus, rank-class preserves date format and ambiguity with $2$ cipher calls per date.

\noindent{\bf\em \sys vs.~alternatives: Date translation.}
Table~\ref{tab:translation} evaluates date translation on $97$ ShareGPT documents. Encipherment shifts dates by $11{,}242$ days, yielding only $53.6\%$ Order, $24.7\%$ Tense, and $0.0\%$ Computed date. \textsc{shift} restores Order and Computed date ($100.0\%$, $97.9\%$) but reduces Tense to $5.2\%$. \textsc{shift-capped} raises Tense to $75.3\%$, while \textsc{shift-today} matches the unprotected prompt ($100.0\%$, $77.3\%$, $97.9\%$). Thus, shifting the stated present preserves temporal relations without sacrificing concealment. Additional NER detector results are in \sref{sec:Detec_choice}.

\vspace{-2mm}
\section{Conclusion}
\vspace{-2mm}

We presented \sys, a user-side PII pseudo-anonymization framework for zero-trust LLM interaction. Combining FPE and keyed semantic surrogates with local restoration, \sys protects plaintext PII while preserving linguistic cues. Across three datasets with 11,322 PII instances, \sys achieves 97.9\% detection recall, 95.2\% protection, and 100\% restoration of echoed surrogates.


\clearpage
\section*{Limitations}

\sys has some limitations. First, its protection depends on PII detection recall: although our detector achieves high recall, implicit, ambiguous, and domain-specific PII may require additional detection mechanisms. Second, \sys relies on the remote LLM preserving surrogate values in its response; values that are omitted, modified, or paraphrased cannot be restored exactly. Third, semantic surrogates preserve entity type and linguistic form but may not preserve every fine-grained contextual relationship, which can affect task-specific applications. Finally, our evaluation focuses on text-based LLM interaction and does not cover multimodal inputs, tool calls, or long-termconversational state, where additional privacy and restoration considerations may arise.



\bibliography{_References}

\appendix
\section*{Appendix}

\section{Additional Results} \label{app:results}

\begin{table*}[!t]
\centering\small
\setlength{\tabcolsep}{4pt}
\begin{tabular}{l|c|c|c|c|c|c|c|c|c}
\toprule
 & \multicolumn{9}{c}{Value Recovery (\%)} \\
\cmidrule(lr){2-10}
Method & \textsc{person} & \textsc{loc} & \textsc{org} & \textsc{email} & \textsc{ssn} & \textsc{phone} & \textsc{card} & \textsc{date} & All \\
\midrule
\multicolumn{10}{l}{\emph{restored from stored state}} \\
\sys &   30.3 &   60.1 &   33.8 &   47.6 &   24.2 &   34.3 &   27.0 &   42.3 &   38.2 \\
Rescriber &   39.4 &   66.7 &   55.4 &   60.7 &   18.2 &   46.3 &   14.3 &   55.7 &   47.0 \\
Pr$\varepsilon\varepsilon$mpt &   33.8 &   62.8 &   35.4 &   35.7 &   15.2 &   46.3 &   22.2 &   42.3 &   39.5 \\
\midrule
\multicolumn{10}{l}{\emph{regenerated by a local model from the original prompt}} \\
InferDPT $\varepsilon\!=\!3$ &   66.9 &   67.8 &   55.4 &   50.0 &   27.3 &   53.7 &   30.2 &   43.6 &   57.9 \\
HaS &   32.7 &   51.9 &   47.7 &   53.6 &   51.5 &   56.7 &   50.8 &   41.6 &   42.6 \\
PAPILLON &   26.4 &   53.0 &   29.2 &   27.4 &  6.1 &   35.8 &   17.5 &   26.2 &   30.5 \\
\midrule
\multicolumn{10}{l}{\emph{leaked across and echoed}} \\
Unprotected &   29.4 &   63.9 &   35.4 &   36.9 &   15.2 &   44.8 &   25.4 &   45.0 &   38.4 \\
CusText $\varepsilon\!=\!3$ &   22.5 &   26.2 &  4.6 &   29.8 &   12.1 &   17.9 &  0.0 &   24.2 &   21.0 \\
DP-Prompt $\varepsilon\!=\!53.4$ &   20.1 &   36.6 &   16.9 &   11.9 &  9.1 &   11.9 &  7.9 &   20.8 &   20.6 \\
SanText+ $\varepsilon\!=\!3$ &   28.6 &   48.1 &   13.8 &  0.0 &  0.0 &  0.0 &  0.0 &   10.1 &   22.1 \\
\midrule
\multicolumn{10}{l}{\emph{nothing crosses, nothing returns}} \\
Staab \etal &  0.9 &  0.0 &  1.5 &  0.0 &  0.0 &  0.0 &  0.0 &  4.0 &  1.0 \\
DP-MLM $\varepsilon\!=\!25$ &  1.7 &  2.2 &  0.0 &  0.0 &  0.0 &  0.0 &  0.0 &  0.0 &  1.1 \\
SEAL &  0.0 &  0.0 &  0.0 &  0.0 &  0.0 &  0.0 &  0.0 &  0.0 &  0.0 \\
\bottomrule
\end{tabular}
\caption{Value recovery by personal-data type for \sys and baselines on the same $1{,}000$ ai4privacy documents using \texttt{Llama-3.2-3B} as the remote LLM.}
\label{tab:valueback}
\end{table*}

\begin{table*}[!t]
\centering\small
\setlength{\tabcolsep}{5pt}
\begin{tabular}{l|l|cc|c|c|c|cc}
\toprule
&& \multicolumn{2}{c|}{Before sending} & & \multicolumn{1}{c|}{After} & & \\
\cmidrule(lr){3-4}\cmidrule(lr){6-6}
\raisebox{2pt}[1pt][2pt]{\rotatebox[origin=c]{90}{\text{Category}}}&Method & Detect & Protect & \shortstack{Untrusted\\model} & Restore & \shortstack{On device\\median (mean)} & \shortstack{Local\\generations} \\
 && (s) & (s) & (s) & (s) & (s) & (\#) \\
\midrule
&Unprotected & -- & -- & 4.43 & -- & 0 & 0 \\
\midrule
\multirow{5}{*}{\rotatebox[origin=c]{90}{\shortstack{Entity\\Substitute}}}
&Pr$\varepsilon\varepsilon$mpt & 0.058 & $<$0.001 & 3.33 & $<$0.001 & 0.059 (0.130) & 0 \\
&Rescriber, our detector & 0.053 & $<$0.001 & 3.13 & $<$0.001 & 0.053 (0.122) & 0 \\
&Rescriber, own detector & 0.114 & $<$0.001 & 1.59 & $<$0.001 & 0.114 (0.385) & 1.1 \\
&HaS & 0.057 & 1.000 & 3.10 & 1.77 & 3.54 (5.78) & 1.9 \\
&
{\bf \sys} & 0.057 & $<$0.001 & 3.40 & $<$0.001 & 0.057 (0.125) & 0 \\
\midrule

\multirow{3}{*}{\rotatebox[origin=c]{90}{\shortstack{Local \\rewrite}}}&PAPILLON & -- & 17.7 & 4.85 & 17.8 & 28.4 (27.8) & 2.0 \\
&Staab \etal & -- & 37.1 & 0.726 & -- & 37.1 (39.3) & 6.0 \\
&SEAL & -- & 75.7 & 0.974 & -- & 75.7 (79.6) & 6.0 \\
\midrule
\multirow{5}{*}{\rotatebox[origin=c]{90}{DP-based}}&
SanText+ $\varepsilon\!=\!3$ & -- & 0.039 & 2.27 & -- & 0.039 (0.069) & 0 \\
&CusText $\varepsilon\!=\!3$ & -- & 0.001 & 2.62 & -- & 0.001 (0.002) & 0 \\
&DP-MLM $\varepsilon\!=\!25$ & -- & 0.053 & 2.31 & -- & 0.053 (0.248) & 0 \\
&DP-Prompt $\varepsilon\!=\!53.4$ & -- & 0.164 & 1.72 & -- & 0.164 (0.545) & 0 \\
&InferDPT $\varepsilon\!=\!3$ & -- & 0.003 & 2.37 & 6.92 & 6.93 (5.97) & 1.0 \\
\bottomrule
\end{tabular}
\caption{Time per request on $50$ PUPA prompts, median over the prompts.}
\label{tab:timing}
\end{table*}



\subsection{ Value recovery: \sys vs.~baselines.} \label{app:permethod} 
Table~\ref{tab:valueback} shows the values returned to the user. \sys recovers
$38.2\%$ of values, matching the unprotected prompt ($38.4\%$) despite
leaking no values; Pr$\varepsilon\varepsilon$mpt and Rescriber recover
$39.5\%$ and $47.0\%$, respectively, as models more often echo their
placeholders (Table~\ref{tab:e2e_stored}). InferDPT ($57.9\%$) and HaS
($42.6\%$) exceed the unprotected prompt because their local models access the
original prompt, while PAPILLON recovers $30.5\%$. SanText+, CusText, and
DP-Prompt recover only values leaked in the clear ($20.6$--$22.1\%$), while
Staab \etal, SEAL, and DP-MLM recover at most $1.1\%$. Thus, \sys combines
near-unprotected value recovery with $<3\%$ leakage, high request retention,
type-valid surrogates, and exact restoration without a local model.

\subsection{Time per Request: \sys vs. baselines}
\label{app:timing}

Table~\ref{tab:timing} times \sys and every baseline on the same $50$ PUPA prompts,
whose median length is $50.5$ words, each running on its own after two
warm-up prompts. Detect and Protect run on the user's device before the prompt
is sent and Restore after the answer returns; the untrusted model is listed for
reference and is not part of the device's cost. \sys spends $0.057$\,s per
request on the device, almost all of it in detection, with encryption and
restoration each under a millisecond and no local generation; its detector and
cipher run on the CPU, so it needs no GPU on the device. The baselines that run
a local model spend about $60$ to $1{,}300$ times as long: HaS
$3.54$\,s, InferDPT $6.93$\,s, PAPILLON $28.4$\,s, Staab \etal\ $37.1$\,s and
SEAL $75.7$\,s. SanText+, CusText, DP-MLM and DP-Prompt stay within
$0.164$\,s, but in Table~\ref{tab:crossed} they leak more than \sys or keep
far less of the request.

\subsection{Local regeneration: \sys vs.~baselines} \label{app:e2e_local}


In Table~\ref{tab:e2e_local}, we report results for local-regeneration methods, where leakage is fixed before inference at $4.4\%$ (HaS) and $0.1\%$ (InferDPT). Their untrusted answers contain at most $3.5\%$ of values, while local regeneration recovers $41.2$--$68.5\%$ and $55.0$--$86.6\%$, respectively, depending on the remote model. This additional recovery is achieved by querying a local model that has access to the original prompt. In contrast, \sys restores echoed values exactly from its local state without querying a local model, while keeping the remote model unaware of the original PII. Thus, \sys avoids the additional local inference cost while providing deterministic restoration for values preserved in the remote response.

\begin{table*}[!t]
\centering\small
\setlength{\tabcolsep}{5pt}
\begin{tabular}{l r c cc cc cc}
\toprule
& & & \multicolumn{2}{c}{\texttt{Llama-3.2-3B}}
& \multicolumn{2}{c}{\texttt{Gemma-2-2B}}
& \multicolumn{2}{c}{\texttt{Phi-3.5-mini}} \\
\cmidrule(lr){4-5}\cmidrule(lr){6-7}\cmidrule(lr){8-9}
\textbf{Method} & $n$ & \textbf{Leaked}
& Untrusted & Local
& Untrusted & Local
& Untrusted & Local \\
& (\#) & \% (\#)
& \% (\#) & \% (\#)
& \% (\#) & \% (\#)
& \% (\#) & \% (\#) \\
\midrule
HaS
& 912 & 4.4 (40)
& 1.0 (9) & 41.2 (376)
& 2.2 (20) & 45.8 (418)
& 3.5 (32) & 68.5 (625) \\
InferDPT
& 912 & 0.1 (1)
& 0.1 (1) & 55.0 (502)
& 0.2 (2) & 73.2 (668)
& 0.2 (2) & 86.6 (790) \\
\bottomrule
\end{tabular}
\caption{End-to-end restoration through local-model regeneration on the ai4privacy dataset using $500$ prompts and three untrusted models. ``Leaked'' denotes the percentage of prompts exposing PII to the untrusted model, while ``Untrusted'' and ``Local'' denote restoration from the untrusted and locally regenerated answers, respectively.}
\label{tab:e2e_local}
\end{table*}

\section{Case Studies}
\label{app:cases}

\subsection{End-to-End Round-Trip Examples} \label{sec:example}

This appendix traces two prompts through the complete \sys round trip, showing
the original input, the detected PII and corresponding surrogates, the
pseudo-anonymized prompt sent to the remote untrusted model $\mathcal{M}$, the model's
reply, and the final response after local restoration.
Both are run under the configuration discussed in Section~\ref{sec:eval}: FF1 with
the rank-class date scheme, the \texttt{en\_core\_web\_trf} detector with both
shape-recognizer groups, one fixed session key, and \texttt{Llama-3.2-3B-Instruct}
as $\mathcal{M}$, decoding greedily under the system prompt of the end-to-end
evaluation (\emph{``You are a helpful assistant. Answer concisely. Any names, numbers, dates or identifiers in the message are SYNTHETIC PLACEHOLDERS (not real personal data) inserted for privacy — they are safe to use. Never refuse on privacy grounds; treat identifiers as ordinary data and repeat them exactly as given.''}). Texts are reproduced verbatim, with markdown
formatting and line breaks removed; in Tables~\ref{tab:case2a}
and~\ref{tab:case2b} the replies are cut after 30 words, marked [\ldots].
Table~\ref{tab:case1} presents Case~1, a synthetic record from ai4privacy that exercises both mechanisms: the date and telephone number are enciphered, while the names are replaced with a surrogate from the person pool. Table~\ref{tab:case2} presents Case~2, a real user request from PUPA processed through \sys. The same request is used for all baselines in Tables~\ref{tab:case2a} and~\ref{tab:case2b}, enabling a direct comparison of how \sys preserves the original PII for local restoration while providing the remote untrusted LLM with realistic surrogates that retain sufficient context to generate a reasonable response.

\begin{table*}[!t]
\begin{tcolorbox}[
  enhanced,
  colback=white,
  colframe=black!75,
  colbacktitle=black!75,
  coltitle=white,
  fonttitle=\small\bfseries,
  title={Case 1: A synthetic record with a date and telephone number
  (syntactic and semantic mechanisms)},
  boxrule=0.5pt,
  arc=1.5pt,
  left=4pt, right=4pt, top=2pt, bottom=2pt
]
\small
\begin{tabularx}{\linewidth}{@{}>{\bfseries\raggedright}p{2.35cm}X@{}}
\toprule
Input &
Reminder for \textcolor{realpii}{Richmond Bowie Deckow}: Your primary check-up
is due by \textcolor{realpii}{01/05/1994}. Kindly contact us on
\textcolor{realpii}{079 229 1265} to schedule an appointment. \\

\midrule
Entities and surrogates &
\textsc{person}~\textcolor{realpii}{Richmond Bowie Deckow}
$\to$ \textcolor{surrogate}{Samantha Leonard} (semantic);
\textsc{date}~\textcolor{realpii}{01/05/1994}
$\to$ \textcolor{surrogate}{08/10/1962} (syntactic);
\textsc{phone}~\textcolor{realpii}{079 229 1265}
$\to$ \textcolor{surrogate}{826 334 6183} (syntactic) \\

\midrule
Sent to $\mathcal{M}$ &
Reminder for \textcolor{surrogate}{Samantha Leonard}: Your primary check-up
is due by \textcolor{surrogate}{08/10/1962}. Kindly contact us on
\textcolor{surrogate}{826 334 6183} to schedule an appointment. \\

\midrule
Reply of $\mathcal{M}$ &
I've noted the reminder for \textcolor{surrogate}{Samantha Leonard}'s primary
check-up, scheduled for \textcolor{surrogate}{08/10/1962}. I will make a note
to contact her on \textcolor{surrogate}{826 334 6183} to schedule an
appointment. \\

\midrule
Restored answer &
I've noted the reminder for \textcolor{realpii}{Richmond Bowie Deckow}'s
primary check-up, scheduled for \textcolor{realpii}{01/05/1994}. I will make
a note to contact her on \textcolor{realpii}{079 229 1265} to schedule an
appointment. \\
\bottomrule
\end{tabularx}
\end{tcolorbox}
\caption{Case~1 (ai4privacy, row 8426). Real values are shown in red and
surrogates in blue. The model echoes all three surrogates, which are restored
exactly; the date preserves its ambiguity and the phone number its grouping. The date \texttt{01/05/1994} admits a day-first and a month-first
reading, and its surrogate \texttt{08/10/1962} admits both as well
(Eq.~\eqref{eq:class}); the telephone number keeps its grouping
(Eq.~\eqref{eq:template}); the name is drawn from the person pool. The model
echoes all three surrogates, and each is restored exactly. The given name
\emph{Richmond} is in neither gendered list, so the bucket falls back to a keyed
choice and yields a female surrogate; the pronoun \emph{her} that the model
chose for it survives restoration.}
\label{tab:case1}
\end{table*}
\begin{table*}[!t]
\centering
\small
\setlength{\tabcolsep}{4pt}
\begin{tcolorbox}[
  enhanced,
  colback=white,
  colframe=black!75,
  colbacktitle=black!75,
  coltitle=white,
  fonttitle=\small\bfseries,
  title={Case 2: A real user request through \sys\ (same request used for all baselines in Tables~\ref*{tab:case2a} and~\ref*{tab:case2b})},
  boxrule=0.5pt,
  arc=1.5pt,
  left=4pt, right=4pt, top=2pt, bottom=2pt
]
\begin{tabularx}{\linewidth}{@{}>{\bfseries\raggedright}p{2.35cm}X@{}}
\toprule
Input &
From the list of following programs: BIDA from \textcolor{realpii}{CMU},
Machine Learning from \textcolor{realpii}{Columbia University}, Data Science
from \textcolor{realpii}{JHU}, AI from \textcolor{realpii}{JHU}, and MBA from
\textcolor{realpii}{JHU}, which are the best compliment programs if I want to
study for 2 Masters programs and to become a technical leader at
\textcolor{realpii}{the Boeing Company}? \\

\midrule
Entities and surrogates &
\textsc{organization}~\textcolor{realpii}{CMU}
$\to$ \textcolor{surrogate}{Dawson-Hale};
\textsc{organization}~\textcolor{realpii}{Columbia University}
$\to$ \textcolor{surrogate}{Franklin, Hooper and Nunez};
\textsc{organization}~\textcolor{realpii}{JHU}
$\to$ \textcolor{surrogate}{Rivera, Thomas and Arroyo}
(3 occurrences, one surrogate);
\textsc{organization}~\textcolor{realpii}{the Boeing Company}
$\to$ \textcolor{surrogate}{Lopez-Melton}. \\

\midrule
Sent to $\mathcal{M}$ &
From the list of following programs: BIDA from
\textcolor{surrogate}{Dawson-Hale}, Machine Learning from
\textcolor{surrogate}{Franklin, Hooper and Nunez}, Data Science from
\textcolor{surrogate}{Rivera, Thomas and Arroyo}, AI from
\textcolor{surrogate}{Rivera, Thomas and Arroyo}, and MBA from
\textcolor{surrogate}{Rivera, Thomas and Arroyo}, which are the best
compliment programs if I want to study for 2 Masters programs and to become
a technical leader at \textcolor{surrogate}{Lopez-Melton}? \\

\midrule
Reply of $\mathcal{M}$ &
To become a technical leader at \textcolor{surrogate}{Lopez-Melton}, I
recommend combining programs that complement each other in terms of skills
and knowledge. Considering your goal, I suggest pairing BIDA from
\textcolor{surrogate}{Dawson-Hale} with Machine Learning from
\textcolor{surrogate}{Franklin, Hooper and Nunez}. BIDA (Business Intelligence
and Data Analysis) will provide a solid foundation in business acumen, data
analysis, and decision-making, which are essential for a technical leader
role. Machine Learning from \textcolor{surrogate}{Franklin, Hooper and Nunez}
will equip you with advanced technical skills in machine learning, data
science, and artificial intelligence, which are highly valued in the industry.
This combination will give you a well-rounded skillset, allowing you to
effectively lead technical teams and make data-driven decisions at
\textcolor{surrogate}{Lopez-Melton}. \\

\midrule
Restored answer &
To become a technical leader at \textcolor{realpii}{the Boeing Company}, I
recommend combining programs that complement each other in terms of skills
and knowledge. Considering your goal, I suggest pairing BIDA from
\textcolor{realpii}{CMU} with Machine Learning from
\textcolor{realpii}{Columbia University}. BIDA (Business Intelligence and
Data Analysis) will provide a solid foundation in business acumen, data
analysis, and decision-making, which are essential for a technical leader
role. Machine Learning from \textcolor{realpii}{Columbia University} will
equip you with advanced technical skills in machine learning, data science,
and artificial intelligence, which are highly valued in the industry. This
combination will give you a well-rounded skillset, allowing you to effectively
lead technical teams and make data-driven decisions at
\textcolor{realpii}{the Boeing Company}. \\
\bottomrule
\end{tabularx}
\end{tcolorbox}
\caption{Case~2 (PUPA, row 158) through \sys. Four organizations are replaced
by semantic surrogates; repeated \texttt{JHU} occurrences share one surrogate.
The reply echoes three surrogates, which are restored exactly.}
\label{tab:case2}
\end{table*}

\begin{table*}[t]
\centering\scriptsize
\setlength{\tabcolsep}{3pt}
\renewcommand{\arraystretch}{1.08}
\begin{CJK}{UTF8}{gbsn}
\begin{tabularx}{\textwidth}{@{}>{\raggedright}p{1.75cm}>{\raggedright\arraybackslash}Xc>{\raggedright\arraybackslash}p{4.05cm}>{\raggedright\arraybackslash}p{4.05cm}@{}}
\toprule
Method & Sent to $\mathcal{M}$ & \shortstack{Leaked\\in sent} & Reply of $\mathcal{M}$ (first 30 words) & Returned to the user (first 30 words) \\
\midrule
\multicolumn{5}{@{}l}{\emph{Unprotected}} \\
Unprotected prompt & From the list of following programs: BIDA from \textcolor{realpii}{CMU}, Machine Learning from \textcolor{realpii}{Columbia University}, Data Science from \textcolor{realpii}{JHU}, AI from \textcolor{realpii}{JHU}, and MBA from \textcolor{realpii}{JHU}, which are the best compliment programs if I want to study for 2 Masters programs and to become a technical leader at \textcolor{realpii}{the Boeing Company}? & 4/4 & To become a technical leader at \textcolor{realpii}{Boeing}, you'll want to focus on programs that complement each other in terms of technical skills and industry relevance. Considering the options, I recommend [\ldots] & \emph{the reply, unchanged} \\
\midrule
\multicolumn{5}{@{}l}{\emph{Per-entity substitution}} \\
\sys & From the list of following programs: BIDA from \textcolor{surrogate}{Dawson-Hale}, Machine Learning from \textcolor{surrogate}{Franklin, Hooper and Nunez}, Data Science from \textcolor{surrogate}{Rivera, Thomas and Arroyo}, AI from \textcolor{surrogate}{Rivera, Thomas and Arroyo}, and MBA from \textcolor{surrogate}{Rivera, Thomas and Arroyo}, which are the best compliment programs if I want to study for 2 Masters programs and to become a technical leader at \textcolor{surrogate}{Lopez-Melton}? & 0/4 & To become a technical leader at \textcolor{surrogate}{Lopez-Melton}, I recommend combining programs that complement each other in terms of skills and knowledge. Considering your goal, I suggest pairing BIDA from \textcolor{surrogate}{Dawson-Hale} [\ldots] & To become a technical leader at \textcolor{realpii}{the Boeing Company}, I recommend combining programs that complement each other in terms of skills and knowledge. Considering your goal, I suggest pairing BIDA [\ldots] \\
Pr$\varepsilon\varepsilon$mpt & From the list of following programs: BIDA from \textcolor{realpii}{CMU}, Machine Learning from \textcolor{realpii}{Columbia University}, Data Science from \textcolor{realpii}{JHU}, AI from \textcolor{realpii}{JHU}, and MBA from \textcolor{realpii}{JHU}, which are the best compliment programs if I want to study for 2 Masters programs and to become a technical leader at \textcolor{realpii}{the Boeing Company}? & 4/4 & To become a technical leader at \textcolor{realpii}{Boeing}, you'll want to focus on programs that complement each other in terms of technical skills and industry relevance. Considering the options, I recommend [\ldots] & To become a technical leader at \textcolor{realpii}{Boeing}, you'll want to focus on programs that complement each other in terms of technical skills and industry relevance. Considering the options, I recommend [\ldots] \\
Rescriber & From the list of following programs: BIDA from [AFFILIATION1], Machine Learning from [AFFILIATION2], Data Science from [AFFILIATION3], AI from [AFFILIATION3], and MBA from [AFFILIATION3], which are the best compliment programs if I want to study for 2 Masters programs and to become a technical leader at [AFFILIATION4]? & 0/4 & To become a technical leader at [AFFILIATION4], I would recommend combining BIDA from [AFFILIATION1] with either Machine Learning from [AFFILIATION2] or Data Science from [AFFILIATION3]. Both Machine Learning and Data [\ldots] & To become a technical leader at \textcolor{realpii}{the Boeing Company}, I would recommend combining BIDA from \textcolor{realpii}{CMU} with either Machine Learning from \textcolor{realpii}{Columbia University} or Data Science from \textcolor{realpii}{JHU}. Both Machine [\ldots] \\
HaS & From the list of following programs: IBM from Harvard, Machine Learning from Stanford University, Data Science from MIT, AI from MIT, and MBA from MIT, which are the best compliment programs if I want to study for 2 Masters programs and to become a technical leader at the SpaceX Company? & 0/4 & To become a technical leader at SpaceX, you'll want to focus on programs that combine technical skills with business acumen. Based on your options, I recommend the following combinations: 1. [\ldots] & 为了成为波音公司的技术领导者，我建议学习以下两门课程： \\
\bottomrule
\end{tabularx}
\end{CJK}
\caption{Case~2 through the control and the per-entity substitution methods.
\emph{Sent to $\mathcal{M}$}: the text that crossed the trust boundary,
verbatim. \emph{Leaked in sent}: gold units of the row that survive verbatim
in that text, the only text $\mathcal{M}$ receives; a partial mention, such as
\emph{Boeing} for the unit \emph{the Boeing Company}, is shown in red but not
counted. \emph{Reply of
$\mathcal{M}$}: what the untrusted model returned. \emph{Returned to the
user}: the answer after the method puts values back on the device, from the
session key and vault for \sys, from its pools and mapping for
Pr$\varepsilon\varepsilon$mpt, from its dictionary for Rescriber, and by its
local seek model for HaS. Real values in this last column were
restored after the reply and never crossed the boundary. Real organization
names in red, \sys's surrogates in blue. HaS's released seek model
was trained on a translation task: from the reply it writes a single Chinese
sentence, \emph{To become a technical leader at Boeing, I recommend studying
the following two courses}, and the recommendations that follow in the reply
are lost.}
\label{tab:case2a}
\end{table*}
\begin{table*}[!t]
\centering\scriptsize
\setlength{\tabcolsep}{3pt}
\renewcommand{\arraystretch}{1.08}
\begin{CJK}{UTF8}{gbsn}
\begin{tabularx}{\textwidth}{@{}>{\raggedright}p{1.75cm}>{\raggedright\arraybackslash}Xc>{\raggedright\arraybackslash}p{4.05cm}>{\raggedright\arraybackslash}p{4.2cm}@{}}
\toprule
Method & Sent to $\mathcal{M}$ & \shortstack{Leaked\\in sent} & Reply of $\mathcal{M}$ (first 30 words) & Returned to the user (first 30 words) \\
\midrule
\multicolumn{5}{@{}l}{\emph{Local rewriting}} \\
PAPILLON & I am interested in studying for two Master's programs and becoming a technical leader at \textcolor{realpii}{Boeing}. Can you recommend the best complement programs from the list of provided options? & 0/4 & To become a technical leader at \textcolor{realpii}{Boeing}, it's essential to have a strong foundation in engineering and business. Based on the provided options, here are some complementary Master's program recommendations: [\ldots] & To become a technical leader at \textcolor{realpii}{Boeing}, consider combining Master's programs in Aerospace Engineering, Business Administration, Systems Engineering, Engineering Management, or Supply Chain Management. Research each program thoroughly to find [\ldots] \\
Staab \etal & From the list of following programs: a technical program from a top university, a Machine Learning program from a prestigious university, a Data Science program from a well-known university, an AI program from a reputable university, and an MBA program from a respected university, which are the best compliment programs if I want to study for 2 Masters programs and to become a technical leader at a major aerospace company? & 0/4 & Based on the requirements, I would recommend the following combinations: 1. Technical Program + AI Program: Technical Program (e.g., from Stanford University) for a strong foundation in engineering and technical [\ldots] & \emph{the reply, unchanged} \\
SEAL & From a selection of options at different organizations, which are the best combinations for someone looking to take on a leadership position in a large group? & 0/4 & Based on various studies and research, here are some of the best combinations for someone looking to take on a leadership position in a large group: 1. Bachelor's degree in [\ldots] & \emph{the reply, unchanged} \\
\midrule
\multicolumn{5}{@{}l}{\emph{Differential privacy}} \\
SanText+ ($\varepsilon{=}3$) & From the lists of following programs : glint from \textcolor{realpii}{CMU} Int Machine Systems from \textcolor{realpii}{Columbia University} Suicide EMAIL Science from \textcolor{realpii}{JHU} , AI from \textcolor{realpii}{JHU} insufficiently and MBA from \textcolor{realpii}{JHU} , specifically are the Best compliment vital if I want to study for 2 Masters programs and to become a technical leader AT \textcolor{realpii}{the Boeing Company} evident & 4/4 & To become a technical leader at \textcolor{realpii}{Boeing}, you'll want to focus on programs that complement each other and provide a strong foundation in technical leadership, business acumen, and industry-specific knowledge. [\ldots] & \emph{the reply, unchanged} \\
CusText ($\varepsilon{=}3$) & am part meaning ONE Following programs: BIDA written \textcolor{realpii}{CMU}, machine Learning Made San University, data biology When \textcolor{realpii}{JHU}, dis Written \textcolor{realpii}{JHU}, AND TESOL according \textcolor{realpii}{JHU}, recently were ONE Best melody Programme give isnt Can Can study In 9 specialization program All let become another SKILLS Leadership Time the Sachs Company? & 2/4 & It seems like you're trying to ask about a program or course that you're interested in. You're referring to the following programs: 1. BIDA (Business Intelligence and Data Analytics) from [\ldots] & \emph{the reply, unchanged} \\
DP-MLM ($\varepsilon{=}25$) & From the options of receiving scenarios: Dna from Nan, Computer Johnson from Ivan University, Market Scientist from Pegasus, Als from Se, and 2007 from Employment, which are the largest available periods if I are to interpret for 2 Md studies and to becoming a formal career at the Shuttle Factory? & 0/4 & To determine the largest available periods, I'll assume you're referring to the number of months. Dna from Nan: 12 months Computer Johnson from Ivan University: 4 years (48 months) Market [\ldots] & \emph{the reply, unchanged} \\
DP-Prompt ($\varepsilon{=}53.4$) & If you want to continue your education and aim to become a technical leader at \textcolor{realpii}{the Boeing Company}, which program should you choose? & 1/4 & To become a technical leader at \textcolor{realpii}{Boeing}, I would recommend pursuing a Master's degree in a field such as: 1. Aerospace Engineering 2. Systems Engineering 3. Management Science and Engineering [\ldots] & \emph{the reply, unchanged} \\
InferDPT ($\varepsilon{=}3$) & peri neither LOT fin assuming Excellent Plate spray discretion Resolver Pattern dynamics marijuana DA interpret is tune person not excitement invisible laundry conditional a talk dal fade Over 655 LAND tot Being puzzle motion boy LINK Han Committee & 0/4 & It appears to be a jumbled collection of words. I'll try to break it down and provide a possible interpretation: "peri" might be related to "peril" or "period" "LOT fin" [\ldots] & Assuming a lotus pattern, I need to exercise discretion as a resolver in a dynamic system, but I'm not excited about the situation, and I'm dealing with invisible laundry that's [\ldots] \\
\bottomrule
\end{tabularx}
\end{CJK}
\caption{Case~2 through the local-rewriting and differential-privacy methods,
with the columns of Table~\ref{tab:case2a}. PAPILLON and InferDPT write the
answer returned to the user with a local model, from the original request and
the reply of $\mathcal{M}$; the other methods return the reply unchanged.}
\label{tab:case2b}
\end{table*}

The system prompt addresses two properties of a safety-tuned model. A
sanitized prompt is dense with identifiers, and the model tends to read such
text as a request to process personal data and declines to answer; and a
surrogate can be restored only if the model repeats it verbatim. We measured
both on $40$ ai4privacy prompts of $120$ to $300$ characters, each holding at
least two entities of the syntactic types and one of the semantic types and no
annotated value outside the eight types of Section~\ref{sec:detect}. Under the
plain instruction \emph{You are a helpful assistant}, the model refuses
$22/40$ sanitized prompts and repeats $31/137$ of their surrogates; under the
system prompt above, the refusals fall to $16/40$ and the repeated surrogates
rise to $52/137$. The unprotected prompts, which carry the real values, show
that the effect is specific to sanitized text: under the plain instruction
they are refused $19/40$ times and $40/137$ of their values are repeated, so
sanitization alone makes the model less cooperative, whereas under the system
prompt they are refused $18/40$ times with $49/137$ values repeated, level with
the sanitized prompts. The system prompt therefore removes a penalty that
sanitization would otherwise impose rather than giving the sanitized prompt an
advantage over the original. It is applied identically to every method and
to the unprotected control. The refusals that remain, between $16/40$ and
$19/40$ under either input, stem from the content of the requests themselves.

\subsection{Detector Choice} \label{sec:Detec_choice}

Table~\ref{tab:detection} shows that shape recognizers improve numeric and date Recall, reaching $95.6\%$ on \textsc{ssn}, \textsc{phone}, and \textsc{card} with \texttt{trf}, while semantic Recall depends on the spaCy model ($84.4\%$--$97.9\%$). Together, they reduce Missed by $31\%$, $64\%$, and $79\%$ for \texttt{sm}, \texttt{lg}, and \texttt{trf}, respectively. The \sys configuration reaches $97.9\%$ Recall, missing only $240$ of $11{,}322$ entities, versus $1{,}941$ for Rescriber and $2{,}677$ for HaS. Thus, detector choice changes Missed by $7.8\times$ and largely determines how much PII remains exposed.

\begin{table}[!t]
\centering
\footnotesize
\setlength{\tabcolsep}{3pt}
\begin{tabular}{c|lrrrrc}
\toprule
 
&  & \multicolumn{4}{c}{Recall (\%)}\\
\cmidrule(lr){3-6}
& Detector & \shortstack{\textsc{Num.}} & \textsc{date} & \textsc{Sem.}
& All & \textsc{Missed(\#)} \\
\midrule
\multirow{4}{*}{%
  \rotatebox[origin=c]{90}{%
    \scriptsize\shortstack{No shape\\[-2pt]recognizers}%
  }}
& \texttt{sm}  & 77.0 & 84.9 & 84.4 & 83.4 & 1{,}875 \\
& \texttt{md}  & 79.0 & 85.9 & 88.3 & 86.6 & 1{,}516 \\
& \texttt{lg}  & 56.0 & 86.4 & 94.6 & 87.9 & 1{,}369 \\
& \texttt{trf} & 42.9 & 98.1 & 97.9 & 90.0 & 1{,}131 \\
\midrule
\multirow{4}{*}{%
  \rotatebox[origin=c]{90}{%
    \scriptsize\shortstack{Both shape\\[-2pt]groups}%
  }}
& \texttt{sm}  & 97.9 & 99.9 & 84.4 & 88.5 & 1{,}301 \\
& \texttt{md}  & 97.6 & 99.9 & 88.3 & 91.2 & 994 \\
& \texttt{lg}  & 97.2 & 99.8 & 94.6 & 95.7 & 489 \\
& \texttt{trf} (\sys) & 95.6 & 99.9 & 97.9 & \textbf{97.9} & \textbf{240} \\
\midrule
\multirow{2}{*}{%
  \rotatebox[origin=c]{90}{%
    \scriptsize\shortstack{Specific\\[-2pt]Detector}%
  }}
& Rescriber & 95.7 & 85.2 & 79.8 & 82.9 & 1{,}941 \\
& HaS       &  2.8 & 96.4 & 87.3 & 76.4 & 2{,}677 \\ [3pt]
 
\bottomrule
\end{tabular}
\caption{Detection on ai4privacy using $10{,}000$ documents with
$11{,}322$ annotated entities. \textsc{Num.} denotes the combined
\textsc{SSN}, \textsc{phone}, and \textsc{card} types.}
\label{tab:detection}
\end{table}

\section{Analysis of Date Translation}
\label{app:translation}

This appendix provides additional details on the translation mechanism introduced in Section~\ref{sec:syntactic}, including the motivating failure (Appendix~\ref{app:tr-instance}), parameter settings (Appendix~\ref{app:tr-params}), admissible-offset construction (Appendix~\ref{app:tr-screen}), the offsets used in the worked example (Appendix~\ref{app:tr-example}), and behavior at range boundaries and for model-computed dates (Appendix~\ref{app:tr-wrap}).

\subsection{Failure of Independent Date Encipherment}
\label{app:tr-instance}

A ShareGPT prompt describes an article written on \texttt{12 August 2014} and updated on \texttt{6 January 2017}, $878$ days later. When the two dates are enciphered independently under the same session key using Eq.~\eqref{eq:rank}, they become \texttt{30 December 1981} and \texttt{4 March 1909}: their separation increases to $26{,}599$ days and their chronological order is reversed. The resulting prompt therefore states that the article was updated seventy-two years before it was written. This failure arises because per-date encipherment preserves the format of each value but not the temporal relationships between values. Our translation mechanism instead applies a common document-level offset, allowing relative temporal structure to be preserved while displacing the original dates.

\subsection{Parameters}
\label{app:tr-params}

Translation operates over the $N=73{,}049$ calendar days from 1900-01-01 to 2099-12-31 and uses two parameters.

\noindent{\bf \em Bound $S$.} $S$ bounds the absolute displacement $|\Delta_d|$ and controls the trade-off between concealment and temporal plausibility. Larger values provide greater displacement from the original date, whereas smaller values keep translated dates closer to their original positions. We use $S=4000$ days, approximately eleven years, yielding $2S+1=8001$ candidate offsets in $[-S,S]$.

\noindent{\bf \em Guard $g$.} $g$ specifies the minimum distance between an offset and the nearest whole-year displacement. It prevents a translation from remaining too close to the original date's position within the calendar year. We use $g=15$ days, requiring an accepted offset to differ by at least fifteen days from any whole-year displacement.

\noindent{\bf \em Pseudorandom function.} $F_K$ in Eq.~\eqref{eq:Off_selec} is HMAC-SHA256 truncated to $64$ bits. Reducing its output modulo $2S+1$ introduces a bias below $2^{-50}$ per candidate, which we neglect.

\subsection{Admissible Offsets}
\label{app:tr-screen}

The admissible set in Eq.~\eqref{eq:Off_selec} is 
\begin{equation} 
\small
\mathcal{A}^\ast = \big\{\, \delta \in [-S, S] :\ 7 \nmid \delta,\ \operatorname{dist}(|\delta|,\, Y\mathbb{Z}) \ge g \,\big\}, \label{eq:screen} 
\end{equation}
where $Y=365.25$ and $\operatorname{dist}(x,Y\mathbb{Z})=\min_{q\in\mathbb{Z}}|x-qY|$ denotes the distance from $x$ to the nearest multiple of $Y$. The first condition excludes offsets that preserve the weekday, including $\delta=0$, while the second excludes offsets that remain close to a whole-year displacement. These constraints remove two forms of calendar structure that can reveal systematic date shifts \citep{evans2023detection}. We illustrate them using the visit date \texttt{2026-07-30}, a Thursday and the $211$th day of the year.

\noindent{\bf \em Weekday.} An offset divisible by $7$ preserves the weekday of every date. For example, adding $7$ days maps \texttt{2026-07-30} to \texttt{2026-08-06}, while adding $3500$ days maps it to \texttt{2036-02-28}; both dates are Thursdays. Because events such as clinical visits are not uniformly distributed across weekdays, preserving the weekday distribution can expose the underlying shift. Excluding multiples of $7$ therefore removes this structure and reduces the candidate space by approximately a factor of seven.

\noindent{\bf \em Season.} An offset close to a whole number of years preserves a date's approximate position within the year and can therefore retain seasonal patterns. For example, adding $365$ days maps \texttt{2026-07-30} to \texttt{2027-07-30}, while adding $730$ days maps it to \texttt{2028-07-29}; both correspond to day $211$ of their respective years. Adding $380$ days yields \texttt{2027-08-14}, which is within $15$ days of a whole-year displacement, whereas $381$ days exceeds the guard. Excluding such offsets disrupts seasonal signatures, such as the concentration of respiratory visits during winter.

\noindent{\bf \em Acceptance rate.} Among integer offsets, one in seven is divisible by $7$, so the weekday constraint retains a fraction $6/7$ of candidates. Over any $Y$-day interval, approximately $2g$ offsets lie within $g$ days of a whole-year displacement, so the season constraint retains a fraction $1-2g/Y$. Approximating the two constraints as independent gives $\Pr[\delta_j\in\mathcal{A}^\ast]\approx\frac{6}{7}\left(1-\frac{2g}{Y}\right)=0.857\times0.918=0.787$. Exhaustive enumeration of all $8001$ offsets in $[-4000,4000]$ gives $6858$ offsets passing the weekday constraint, $7376$ passing the season constraint, and $6322$ satisfying both, for an empirical acceptance rate of $0.790$. Thus, the first accepted draw has index $j^\ast=0$ for approximately $79\%$ of documents, with expected draw index $1/0.790\approx1.27$.

\begin{table}[t]
\centering\scriptsize
\setlength{\tabcolsep}{3.5pt}
\begin{tabular}{llll}
\toprule
& \textsc{shift} & \textsc{shift-capped} & \textsc{shift-today} \\
\midrule
offset $\Delta_d$ (days) & $-3{,}328$ & $-27$ & $-3{,}328$ \\
\midrule
birth \texttt{1958-04-12} & \texttt{1949-03-02} & \texttt{1958-03-16} & \texttt{1949-03-02} \\
visit \texttt{2026-07-30} & \texttt{2017-06-19} & \texttt{2026-07-03} & \texttt{2017-06-19} \\
follow-up \texttt{2026-09-15} & \texttt{2017-08-05} & \texttt{2026-08-19} & \texttt{2017-08-05} \\
\midrule
present stated to $\mathcal{M}$ & \texttt{2026-08-13} & \texttt{2026-08-13} & \texttt{2017-07-03} \\
\midrule
order, duration & \checkmark & \checkmark & \checkmark \\
follow-up still pending & $\times$ & \checkmark & \checkmark \\
smallest date movement & $3{,}328$\,d & $27$\,d & $3{,}328$\,d \\
\bottomrule
\end{tabular}
\caption{\textsc{shift} and its two tense remedies on one clinical document,
under one session key, with the true present at \texttt{2026-08-13}. All three preserve
order and duration. \textsc{shift}
moves the follow-up into the past; \textsc{shift-capped} keeps its tense but confines
the offset to $27$ days; \textsc{shift-today} keeps both by moving the stated
present with the data.}
\label{tab:temporal}
\end{table}

\subsection{Offsets of the Worked Example}
\label{app:tr-example}

For the document of Table~\ref{tab:temporal} the candidate interval is
$[-4000, 4000]$.

\paragraph{\textsc{shift} and \textsc{shift-today}.} The first draw gives
$F_{\Ksyn}(\tau_d \Vert 0) \bmod 8001 = 672$, hence
$\delta_0 = -4000 + 672 = -3{,}328$, nine years and $41$ days. It is neither
a multiple of $7$ nor within $g$ of a whole year, so $j^\ast = 0$ and
$\Delta_d = -3{,}328$, which moves the visit to a Monday and from day $211$ to
day $170$ of its year. Both methods use this offset, and \textsc{shift-today}
also states the present as $\tilde{t}_0 = $ \texttt{2026-08-13}
$-\,3{,}328$ days $=$ \texttt{2017-07-03}.

\paragraph{\textsc{shift-capped}.} The latest past date is the visit, $14$
days before $t_0$, and the earliest future date is the follow-up, $33$ days
after it, so Eq.~\eqref{eq:capped} gives $W_d = [-32, 13]$. Of its $46$
offsets, $16$ are admissible, since the guard removes every $|\delta| < 15$
and the weekday test the multiples of $7$; the first admissible draw is
$-27$.

The smallest displacement any date receives is therefore $3{,}328$ days under
\textsc{shift} and \textsc{shift-today}, and $27$ days under
\textsc{shift-capped}.

\subsection{Wrap-Around and Computed Dates}
\label{app:tr-wrap}

\paragraph{Wrap-around.} The modulus in Eq.~\eqref{eq:same_off} guarantees a
valid surrogate and an exact inverse for every document. A date whose shifted
day number leaves $[0, N)$, however, wraps to the far end of the range and
loses its order and distance relative to the other dates of the document,
though not its restoration. No date wraps when $0 \le a_i + \Delta_d < N$ for
all $i$; since $|\Delta_d| \le S$, this holds whenever every date lies at
least $S$ days inside the range, between 1910-12-15 and 2089-01-17. The
synthetic ai4privacy corpus draws dates from the whole range, and two of the
fourteen dates in its multi-date documents, from 1902 and 1907, fall outside
this interval: an offset that carries them past 1900-01-01 wraps them into
the 2090s, where they restore exactly but lose their relations to the rest of
the document.

\paragraph{Computed dates.} Because a rotation commutes with addition, the
inverse also carries back dates the model computed rather than copied. If
$\mathcal{M}$ returns a date $c$ days after the surrogate $\tilde{v}_i$,
Eq.~\eqref{eq:untranslate} yields
\[
\nu^{-1}\big((\nu(\tilde{v}_i) + c - \Delta_d) \bmod N\big)
\;=\; \nu^{-1}(a_i + c)
\]
whenever the result does not wrap, which is $c$ days after the real date. On
the document of Table~\ref{tab:temporal}, thirty days after the surrogate
visit is \texttt{2017-07-19}, which inverts to \texttt{2026-08-29}, thirty
days after the real visit. The inverse of Eq.~\eqref{eq:rank} maps the same
\texttt{2017-07-19} to \texttt{1973-09-30}, since a permutation preserves
neither order nor distance.

\section{Privacy Guarantee}
\label{app:privacy}

We state the privacy guarantee of \sys as a game in which an adversary must
tell apart the sanitized versions of two prompts that agree on a
\emph{leakage function}: the information that surrogates reveal by design,
which is what keeps them useful to the model.

\paragraph{Setting.} The adversary $\mathsf{A}$ is the honest-but-curious
provider of Section~\ref{sec:threat}. It sees every sanitized prompt and
knows the design of \sys, but not $\Km$, the session keys or the vault. Its
replies are computed from what it has seen, so they add nothing to its view,
and restoration is local. We treat the detector as a black box: the guarantee
covers the PII entities it returns. \sys relies on three cryptographic
primitives: HKDF derives its keys, FF1 enciphers syntactic values, and
HMAC-SHA256 selects semantic surrogates and translation offsets. Each keyed
use has its own key: $\Ksyn$ keys FF1 and the translation offsets, and $K_{\mathrm{sem}}$ the semantic hash
(Eq.~\eqref{eq:trackkeys}).

\paragraph{Leakage.} For the prompts $X$ of one session, $\Leak(X)$ consists
of (i) the text outside the transformed PII entities, including PII the
detector missed, and the position and type of each entity; (ii) the format
each surrogate keeps, such as the template of an identifier, the layout and
ambiguity class of a date, and the gender and token count of a name;
(iii) which entities of the session are the same value; and (iv) under
translation, the differences between the dates of each document. $\Leak$
contains nothing that relates values in different sessions.

\paragraph{Privacy game.} In the game of Figure~\ref{fig:privgame},
$\mathsf{A}$ opens sessions and, for each prompt, submits two candidates
$x_0, x_1$ with the same leakage. A secret random bit $b$, drawn once at the
start, decides which candidate is sanitized: $\mathsf{A}$ always receives the
sanitized $x_b$, and at the end outputs a guess $b'$. It may choose each pair
after seeing earlier answers, as a provider steering a conversation could.
$\mathsf{G}_{\Leak}(\mathsf{A})$ denotes one run of the game, which outputs
$1$ when $b' = b$, and $\Pr[\mathsf{G}_{\Leak}(\mathsf{A}) = 1]$ is the
probability of this over the keys, $b$ and the coins of $\mathsf{A}$. Since
guessing at random already wins with probability $1/2$, the advantage
rescales it so that $0$ means no better than guessing and $1$ a certain win:
\[
\Adv_{\sys}(\mathsf{A}) = 2\Pr[\mathsf{G}_{\Leak}(\mathsf{A}) = 1] - 1 .
\]
A small advantage means that sanitized prompts reveal nothing about the PII
beyond $\Leak$; in particular, $\mathsf{A}$ recovers a protected value no
better than it could from $\Leak$ alone.

\begin{figure}[!t]
\centering\small
\begin{tcolorbox}[enhanced,boxrule=0.4pt,arc=1pt,left=4pt,right=4pt,top=3pt,bottom=3pt,colback=white]
\textsc{Initialize}: $\Km \xleftarrow{\$} \{0,1\}^{\kappa}$;\quad
$b \xleftarrow{\$} \{0,1\}$\\[3pt]
\textsc{NewSession}: open session $s$ with a fresh $\mathtt{sid}_s$,
$\Ks = \mathrm{HKDF}(\Km, \mathtt{sid}_s)$, an empty vault $V_s$ and empty
transcripts $X^s_0, X^s_1$\\[3pt]
\textsc{Sanitize}$(s, x_0, x_1)$:\\
\hspace*{1em}if $\Leak(X^s_0 \Vert x_0) \neq \Leak(X^s_1 \Vert x_1)$:
return $\bot$\\
\hspace*{1em}append $x_0$ to $X^s_0$ and $x_1$ to $X^s_1$\\
\hspace*{1em}return $\San_{\Ks}(x_b)$, updating $V_s$\\[3pt]
\textsc{Finalize}$(b')$: return $[\,b' = b\,]$
\end{tcolorbox}
\caption{The privacy game $\mathsf{G}_{\Leak}$.}
\label{fig:privgame}
\end{figure}

\begin{theorem}
\label{thm:privacy}
If \sys enciphers dates rather than translating them, then for every
adversary $\mathsf{A}$ that opens $\sigma$ sessions,
\[
\begin{aligned}
\Adv_{\sys}(\mathsf{A}) \le 2(\sigma+1)\big(&\Adv_{\mathrm{HKDF}}
  + \Adv_{\mathrm{FF1}}\\ &+ \Adv_{\mathrm{HMAC}}\big),
\end{aligned}
\]
where $\Adv_{\mathrm{HKDF}}$, $\Adv_{\mathrm{FF1}}$ and
$\Adv_{\mathrm{HMAC}}$ are the best advantages, in about the time of
$\mathsf{A}$, in telling HKDF, FF1 and
HMAC-SHA256 under a random key from their ideal versions: a random function
for HKDF and HMAC-SHA256, and an independent random permutation per tweak and
width for FF1.
\end{theorem}

\begin{proof}[Proof sketch]
We compare the real game with an \emph{ideal} game in which the cryptography
is replaced by truly random objects.

\emph{Step 1: from real to ideal.} In every session, replace the keys
derived by HKDF with random keys, FF1 with a random permutation for each
tweak and width, and HMAC with a random function. These primitives are
designed so that no efficient adversary can tell them from these ideal
versions. If a replacement changed how often $\mathsf{A}$ wins,
$\mathsf{A}$ could be used to do exactly that, so each replacement changes
the chance that $\mathsf{A}$ wins by at most the corresponding advantage,
which is negligible for secure primitives. Each session has its own keys, so
the replacement is made session by session, which gives the factor
$(\sigma+1)$.

\emph{Step 2: the ideal game reveals only $\Leak$.} In the ideal game a
surrogate is chosen by chance rather than computed from the real value, so it
carries no information about it. A random permutation sends each new value to
a fresh random output, so an enciphered surrogate is a random value in the
format of the original, repeated only when the value repeats; cycle walking
keeps this true for dates \cite{black2002ciphers}. For example,
\texttt{479-13-8850} and \texttt{123-45-6789} both become a random
nine-digit number in the layout \texttt{\_\_\_-\_\_-\_\_\_\_}. A random
function selects a pool entry for each new value, so a semantic surrogate is
a random unused entry of its bucket, again repeated only when the value
repeats: \texttt{Marie Curie} and \texttt{Anna Novak} both become a random
female two-token name. What carries over is therefore only what $\Leak$
contains: the text around the PII, the format of each value, and which values
repeat.

\emph{Conclusion.} The two candidates have the same $\Leak$, so in the ideal
game their sanitized versions have the same distribution and $\mathsf{A}$
wins with probability exactly $1/2$. By Step~1, $\mathsf{A}$ therefore wins
the real game with probability at most $1/2$ plus the total cost of the
replacements. Since the advantage is
$2\Pr[\mathsf{G}_{\Leak}(\mathsf{A}) = 1] - 1$, it is at most twice that
cost, which is the bound.
\end{proof}

\begin{corollary}[Unlinkability]
\label{cor:unlink}
Replacing a value in one session by another value of the same type and
format that appears nowhere else in that session leaves $\Leak$ unchanged, so
$\mathsf{A}$ cannot tell whether two sessions mention the same person with
advantage above the bound of Theorem~\ref{thm:privacy}.
\end{corollary}

The corollary rests on the per-session key: under one long-lived key, a value
would receive the same surrogate in every session.

\paragraph{Date translation.} Translation moves every date of a document by
the same secret offset $\Delta_d$ of at most $S$ days (Eq.~\eqref{eq:Off_selec}).
It therefore reveals the gaps between the dates of a document, which lets the
model reason about intervals, and hides only where the document lies on the
calendar: each real date lies within $S$ days of its surrogate (about
$\pm 11$ years for $S = 4000$), all positions in that range being equally
likely. The admissible set of Eq.~\eqref{eq:screen} adds one clue: since no
offset is a multiple of $7$ or near a whole number of years, a real date never
falls on its surrogate's weekday or within $g$ days of its day of the year.
Each further document about the same timeline adds evidence. Tense-preserving
variants reveal more: \textsc{shift-today} states the shifted present
$t_0 + \Delta_d$, so a provider that knows the real date $t_0$ recovers
$\Delta_d$ and every date. Encipherment is therefore the default, and
translation is used only when intervals between dates matter.

\section{Examples of the Syntactic Mechanism}\label{app:worked}

Table~\ref{tab:worked} traces the syntactic mechanism of
Section~\ref{sec:syntactic} on one value of each kind under a single session
key: the template and digit string of direct encipherment, and the rank,
domain size and cipher calls of rank-then-encipher for each date
representation, and two dates the cipher declines, which the semantic
mechanism replaces.

\noindent{\bf\em Examples on how Encipherment Mechanism encrypt the PII}The identifier \texttt{479-13-8850}
splits into the template \texttt{\_\_\_-\_\_-\_\_\_\_} and the digit string
\texttt{479138850}; FF1 maps the digits to \texttt{790880367}, and refilling
the template yields \texttt{790-88-0367}. The grouped card number
\texttt{4012-8888-8888-1881} and its bare form \texttt{4012888888881881} share
the digit string \texttt{4012888888881881}, so both encipher to the same digits
\texttt{2836715022441866} and differ only in the template that refills them.
For telephone numbers a leading country code is excluded from $\delta(v)$ and
copied verbatim: \texttt{+1-202-555-0173} keeps the \texttt{+1}, enciphers
\texttt{2025550173} to \texttt{5371444801}, and refills as
\texttt{+1-537-144-4801}, staying in its country and layout.

\noindent{\bf\em Examples of Date Encipherment}
\begin{itemize}[leftmargin=*]

\item \textbf{Fixed-layout dates.} For a three-part date with a four-digit year, \sys parses and ranks the date within its valid calendar domain, then renders the surrogate using the original field order, widths, and separators. For example,
\texttt{1985-07-24} becomes \texttt{1946-10-01}, while
\texttt{12/05/2001} may become \texttt{08/17/1974}, preserving the
\texttt{MM/DD/YYYY} layout. Similarly,
\texttt{31-01-2020} can be transformed to another valid
\texttt{DD-MM-YYYY} date without changing the representation.

\item \textbf{Ambiguous numeric dates.} For dates that admit multiple interpretations, \sys ranks the value within its ambiguity class rather than assuming a locale. For example, \texttt{10/03/1976}, which can represent either 10~March or 3~October, becomes \texttt{06/07/2033}, which remains valid under both interpretations. A day-first-only date such as \texttt{25/12/1999} can become \texttt{23/02/1916}, preserving its day-first interpretation. Likewise, an unambiguous year-first representation such as \texttt{1976/03/10} remains within the corresponding \textsc{ymd} class after transformation.

\item \textbf{Textual dates.} For dates written using month names, \sys canonicalizes the date before rank-based encipherment and renders the surrogate in the original textual style. For example, \texttt{February 4, 1909} becomes \texttt{October 31, 1966}, while \texttt{Feb 4th, 1909} becomes \texttt{Oct 31st, 1966}. Similarly, \texttt{September 12, 1987} can become \texttt{March 28, 2014}, and \texttt{12 September 1987} can be rendered as another valid date while retaining the original day--month--year textual ordering.

\item \textbf{Timestamps.} For timestamps, \sys independently transforms the date and time-of-day components using separate domains and tweaks while preserving the timestamp structure. For example,
\texttt{1902-08-21T22:37:28.285Z} becomes
\texttt{2012-10-10T10:06:55.270Z}. A timestamp without fractional
seconds, such as \texttt{2024-06-15T14:32:10Z}, retains the same precision,
while \texttt{2023-11-02T08:15:42.50+00:00} preserves both its fractional
precision and time-zone offset.

\item \textbf{Dates the cipher declines.} A fragment with no full calendar
date to rank, such as \texttt{20th September} (two digits) or \texttt{1/20}
(three), has no layout with a four-digit year to parse and too few digits to
encipher directly, and neither has a time of day or an age that the detector
labels as a date, such as \texttt{18:13} or \texttt{53 years old}. The
syntactic mechanism declines each such string and hands it to the semantic
mechanism (Section~\ref{sec:semantic}), which issues a surrogate from the
session vault. A partial date keeps its shape and stays a valid day and
month: \texttt{20th September} may become \texttt{10th October} and
\texttt{1/20} may become \texttt{9/29}. Any other string keeps its letters and
separators while each digit is replaced under the session key, so
\texttt{18:13} may become \texttt{40:63}, which hides the value but is not a
valid time. Each surrogate is restored by vault lookup. An expression with no
digit, such as \emph{next week}, has nothing to replace and is sent
unchanged.
\end{itemize}

\begin{table*}[t]
\centering\scriptsize
\setlength{\tabcolsep}{2.5pt}
\resizebox{\linewidth}{!}{
\begin{tabular}{lllll}
\toprule
Input & Decomposition & Rank $\rho$ / $n$ & FF1 output (walks) & Surrogate \\
\midrule
\multicolumn{5}{l}{\emph{Direct encipherment of the digit string}} \\
\texttt{479-13-8850} & $\pi$ \texttt{\_\_\_-\_\_-\_\_\_\_}, $\delta$ \texttt{479138850} & --- & \texttt{790880367} & \texttt{790-88-0367} \\
\texttt{4012-8888-8888-1881} & $\delta$ \texttt{4012888888881881} & --- & \texttt{2836715022441866} & \texttt{2836-7150-2244-1866} \\
\texttt{4012888888881881} & same $\delta$, no separators & --- & \texttt{2836715022441866} & \texttt{2836715022441866} \\
\texttt{+1-202-555-0173} & \texttt{+1} kept, $\delta$ \texttt{2025550173} & --- & \texttt{5371444801} & \texttt{+1-537-144-4801} \\
\midrule
\multicolumn{5}{l}{\emph{Rank-then-encipher over the valid domain}} \\
\texttt{1985-07-24} & layout $4{\times}2{\times}2$, class \textsc{ymd} & 31,250 / 73,049 & 17,074 (1) & \texttt{1946-10-01} \\
\texttt{10/03/1976} & layout $2{\times}2{\times}4$, class \textsc{both} & 11,054 / 28,800 & 19,218 (6) & \texttt{06/07/2033} \\
\texttt{25/12/1999} & layout $2{\times}2{\times}4$, class \textsc{dmy} & 22,058 / 44,249 & 3,660 (3) & \texttt{23/02/1916} \\
\texttt{February 4, 1909} & canonical \texttt{1909-02-04}, all days & 3,321 / 73,049 & 24,409 (1) & \texttt{October 31, 1966} \\
\texttt{Feb 4th, 1909} & canonical \texttt{1909-02-04}, all days & 3,321 / 73,049 & 24,409 (1) & \texttt{Oct 31st, 1966} \\
\texttt{1902-08-21T22:37:28.285Z} & date part, all days & 962 / 73,049 & 41,190 (3) & \texttt{2012-10-10} \\
\quad (same instant) & time part, ms of day & 81,448,285 / 86,400,000 & 36,415,270 (1) & \texttt{T10:06:55.270Z} \\
\midrule
\multicolumn{5}{l}{\emph{Declined by the cipher, replaced from the vault}} \\
\texttt{20th September} & no four-digit year; 2 digits & --- & declined & \texttt{28th December} (vault) \\
\texttt{18:13} & time of day; 4 digits & --- & declined & \texttt{48:01} (vault) \\
\bottomrule
\end{tabular}}
\caption{Worked examples of the syntactic mechanism under FF1, one session
key. $\pi$ and $\delta$ are the template and digit string of
Eq.~\eqref{eq:template}; $\rho$ and $n$ the rank and domain size of
Eq.~\eqref{eq:rank}; \emph{walks} the cipher calls until the output landed
inside the domain; \emph{(vault)} a declined value the semantic mechanism
replaces.}
\label{tab:worked}
\end{table*}

\section{Examples of the Semantic Mechanism}\label{app:pools}

The semantic pools of Section~\ref{sec:semantic} are built from the locale
data providers of the Faker library \cite{faraglia_faker}, version 40.38.0. Each
locale provider ships lists of male given names, female given names,
surnames and, for some locales, cities, compiled from public sources; the
United States lists, for instance, derive from Social Security
Administration given-name tables and Census surname counts.
Table~\ref{tab:pools} lists the fifteen Latin-script locales we draw from and
the number of entries each contributes. The three name lists are kept per
locale, so that a surrogate person draws given name and
surname from one locale, and are also merged into global pools in which a
name occurring in several locales is counted once; the merged pools are used for
e-mail addresses. Cities are not kept per locale: the city lists of all locales
are merged into a single pool of $21{,}396$ entries, of which $16{,}529$ come
from the Italian list and $169$ from a curated list of world cities. The
company pool holds $5{,}392$ distinct names generated from company-name
patterns over the surname lists with a fixed seed, and the e-mail domain pool
holds $15$ common providers. Table~\ref{tab:semantic} lists the surrogates
that these pools yield for six values under one session key.

\begin{table}[t]
\centering\scriptsize
\setlength{\tabcolsep}{3pt}
\begin{tabular}{@{}llrrr@{}}
\toprule
Locale & Language / region & Male names & Female names & Surnames \\
\midrule
\texttt{cs\_CZ} & Czech / Czechia            &   127 &   137 &   200 \\
\texttt{da\_DK} & Danish / Denmark           &   244 &   252 &    91 \\
\texttt{de\_DE} & German / Germany           &   999 & 1,003 &   404 \\
\texttt{en\_GB} & English / United Kingdom   &   182 &   182 &   500 \\
\texttt{en\_US} & English / United States    &   322 &   381 & 1,000 \\
\texttt{es\_ES} & Spanish / Spain            &   473 &   478 & 1,085 \\
\texttt{fi\_FI} & Finnish / Finland          &   199 &   199 &   400 \\
\texttt{fr\_FR} & French / France            &    88 &   127 &   400 \\
\texttt{hu\_HU} & Hungarian / Hungary        &   100 &   100 &   100 \\
\texttt{it\_IT} & Italian / Italy            &   273 &   185 & 1,164 \\
\texttt{nl\_NL} & Dutch / Netherlands        &   251 &   251 & 1,070 \\
\texttt{no\_NO} & Norwegian / Norway         &   100 &   100 &   101 \\
\texttt{pt\_BR} & Portuguese / Brazil        &   142 &   140 &   102 \\
\texttt{ro\_RO} & Romanian / Romania         &   265 &   362 &    39 \\
\texttt{sv\_SE} & Swedish / Sweden           &   500 &   500 &   500 \\
\midrule
\multicolumn{2}{@{}l}{merged, duplicates removed} & 2,969 & 2,750 & 6,337 \\
\bottomrule
\end{tabular}
\caption{Locales feeding the semantic pools and the given names and surnames
each contributes. The city pool ($21{,}396$), company pool ($5{,}392$) and
e-mail domain pool ($15$) are merged across locales.}
\label{tab:pools}
\end{table}

\begin{table}[t]
\centering\scriptsize
\setlength{\tabcolsep}{3pt}
\begin{tabular}{@{}lll@{}}
\toprule
Type & Input $v$ & Surrogate $\tilde{v}$ \\
\midrule
\textsc{person} & \texttt{Marie Curie} & \texttt{Ágnes Orosz} \\
\textsc{person} & \texttt{MARIE~~CURIE} & \texttt{Ágnes Orosz} \\
\textsc{person} & \texttt{Pierre Curie} & \texttt{Carlo Krogh} \\
\textsc{person} & \texttt{Marie} & \texttt{Taru} \\
\textsc{loc} & \texttt{Paris} & \texttt{Paesana} \\
\textsc{org} & \texttt{Sorbonne University} & \texttt{Hamilton-Williams} \\
\textsc{email} & \texttt{marie.curie@sorbonne.fr} & \texttt{piersanti.madden@proton.me} \\
\bottomrule
\end{tabular}
\caption{Semantic surrogates issued under one session key. A re-cased
repetition shares its surrogate; under a second session key
\texttt{Marie Curie} becomes \texttt{Lucie Musil}.}
\label{tab:semantic}
\end{table}

\section{Experimental Setup Details} \label{app:setup}

\subsection{Datasets} \label{app:datasets}

We evaluate \sys on four datasets covering synthetic, real-world, and
benchmark-specific text:

\begin{itemize}[leftmargin=*]
\item {\bf \emph{ai4privacy} \cite{ai4privacy_pii_masking_200k}.} Template-generated
text with gold PII annotations. Among the three general-purpose corpora,
it is the only annotated dataset, so we use it to evaluate detection
recall and per-entity protection.
 
\item {\bf \emph{WildChat} \cite{zhao2024wildchat}.} Raw user--assistant
conversations collected from real-world interactions. It evaluates
whether results observed on template-generated text transfer to naturally
occurring prompts.

\item {\bf \emph{ShareGPT} \cite{sharegpt2023vicuna}.} User-shared
conversations from real-world interactions. Together with WildChat, it
tests generalization to naturally occurring prompts while introducing
different selection biases, as users are less likely to share
conversations containing their own personal data.

\item {\bf \emph{PUPA} \cite{siyan2025papillon}.} The benchmark introduced
with PAPILLON, containing $664$ WildChat turns annotated with gold
personal-data units. We use PUPA to evaluate privacy--utility trade-offs
consistently across all compared systems.

\end{itemize}

\subsection{Detector Configuration} \label{app:detector}

Unless stated otherwise, the detector uses the transformer sequence model (\texttt{en\_core\_web\_trf}) with both shape-recognizer groups enabled. The analyzer in Section~\ref{sec:detect} combines one spaCy model with validating recognizers for credit-card numbers (Luhn), SSNs, telephone numbers, and e-mail addresses. We add two independently toggleable groups of shape recognizers to capture surface forms that the sequence model may miss on multi-locale text. Unlike the validating recognizers, these are regular-expression patterns that match surface form alone and perform no validation: card patterns match any digit sequence of the required length regardless of Luhn validity, while the eleven-digit identifier pattern matches any sequence of that length. Each pattern is assigned a confidence score in $[0,1]$, shown beside it in Table~\ref{tab:shapes}, which the analyzer uses in two ways. First, matches below a threshold are discarded; we set this threshold to $0.4$, which admits all eight patterns, so increasing it prunes patterns from the bottom, with the year-less numeric date pattern ($0.45$) removed first. Second, when two recognizers return the same entity, the lower-scoring match is discarded, while an entity contained within another is discarded regardless of score. The scores are therefore relative: a pattern without separators, which may match ordinary numbers, receives a lower score than one whose punctuation makes false matches less likely.


\begin{table}[!t]
\centering\scriptsize
\setlength{\tabcolsep}{4pt}
\resizebox{\linewidth}{!}{
\begin{tabular}{@{}llr@{}}
\toprule
entity & pattern & score \\
\midrule
\multicolumn{3}{@{}l@{}}{\emph{Identifier shapes}}\\
\textsc{card}  & \texttt{\textbackslash b\textbackslash d\{13,19\}\textbackslash b} & 0.50 \\
\textsc{card}  & \texttt{\textbackslash b(\textbackslash d\{4\}[\ -])\{3\}\textbackslash d\{4\}\textbackslash b} & 0.60 \\
\textsc{ssn}   & \texttt{\textbackslash b\textbackslash d\{3\}.\textbackslash d\{4\}.\textbackslash d\{4\}\textbackslash b} & 0.60 \\
\textsc{ssn}   & \texttt{\textbackslash b\textbackslash d\{11\}\textbackslash b} & 0.50 \\
\textsc{ssn}   & \texttt{\textbackslash b\textbackslash d\{3\}-\textbackslash d\{2\}-\textbackslash d\{4\}\textbackslash b} & 0.60 \\
\textsc{phone} & \texttt{\textbackslash +\textbackslash d[\textbackslash d\ .\textbackslash -]\{8,16\}\textbackslash d} & 0.55 \\
\addlinespace[2pt]
\multicolumn{3}{@{}l@{}}{\emph{Date shapes}}\\
\textsc{date}  & \texttt{\textbackslash b\textbackslash d\{1,2\}(st|nd|rd|th)\textbackslash s+MONTH\textbackslash b} & 0.60 \\
\textsc{date}  & \texttt{\textbackslash b\textbackslash d\{1,2\}/\textbackslash d\{2,4\}\textbackslash b} & 0.45 \\
\bottomrule
\end{tabular}}
\caption{The shape recognizers, as registered with the analyzer.
\texttt{MONTH} abbreviates the alternation over full and three-letter month
names.}
\label{tab:shapes}
\end{table}

 \subsection{Baselines}
\label{app:baselines}

We provide additional details on the baselines spanning entity substitution, local prompt rewriting, and differential-privacy-based protection. These categories differ in whether they operate on detected entities, rewrite the prompt more broadly, or perturb its content to provide privacy.

\vspace{1.5mm}
\noindent\textbf{\em 1. Entity substitution.} These methods explicitly identify personal-data
entities and replace them with encrypted values, typed placeholders, or
surrogates before sending the request to the remote model.

\begin{itemize}[leftmargin=*]
\item {\bf \emph{Pr$\varepsilon\varepsilon$mpt} \cite{chowdhury2025pr}.} It enciphers detected entities with FPE. It natively supports names; for numeric entity types, we apply its FPE primitive to the entities detected by our detector.
\item {\bf \emph{Rescriber} \cite{zhou2025rescriber}.} It replaces entities with typed placeholders such as \texttt{[NAME1]}, selected by a prompted 8B detector, and restores them using a per-conversation dictionary.

\item {\bf \emph{HaS} \cite{chen2023hide}.} It uses a fine-tuned local model to hide entities and a second local model to restore them in the response.
\end{itemize}
For these methods, the entity-level transformation is applied to the same
entities detected by our detector, allowing us to isolate differences in
transformation and restoration rather than detection coverage.

\vspace{1.5mm}
\noindent\textbf{\em 2. Local rewrite.} These methods protect personal information by rewriting
the request with a trusted local model rather than replacing individual
detected entities.
\begin{itemize}[leftmargin=*]
\item {\bf \emph{PAPILLON} \cite{siyan2025papillon}.} It uses a trusted local model to rewrite the request; we reproduce it using the authors' released pipeline and optimized prompts.
\item {\bf \emph{Staab \etal} \cite{staab2402large}.} The adversarial anonymizer alternates between an inference model that predicts personal attributes and a rewriter that modifies the prompt to defeat those predictions.

\item {\bf \emph{SEAL} \cite{kim2026self}.} It distills this adversarial loop into a single local 8B model, which we train using the authors' recipe because pretrained weights are not released.
\end{itemize}
None of these methods sends per-entity surrogates or relies on an explicit PII
detector; instead, the local model rewrites the request as a whole.

\vspace{1.5mm}
\noindent\textbf{\em 3. Differential privacy.} These methods provide formal or approximate
privacy by perturbing tokens, embeddings, or the entire prompt rather than
targeting PII entities explicitly.
\begin{itemize}[leftmargin=*]
\item {\bf \emph{SanText+} \cite{yue2021differential} \& \emph{CusText} \cite{chen2023customized}.} replace tokens using metric local differential privacy over an embedding space.
\item {\bf \emph{DP-MLM} \cite{meisenbacher2024dp}.} It resamples each word from a masked language model conditioned on its context.
\item {\bf \emph{DP-Prompt} \cite{utpala2023locally}.} It regenerates the entire document, with the sampling temperature determined by $\varepsilon$.
\item {\bf \emph{InferDPT} \cite{tong2025inferdpt}.} It perturbs the prompt and uses a local extraction model to recover information from the generated response.
\end{itemize}
None of these methods uses an explicit PII detector; instead, each operates on
the input at the token level or rewrites the full prompt. We evaluate each
method at its reported operating point: $\varepsilon=3$ for SanText+, CusText,
and InferDPT, $\varepsilon=25$ for DP-MLM, and $\varepsilon=53.4$ for DP-Prompt,
corresponding to its central sampling temperature. For DP-MLM and DP-Prompt,
$\varepsilon$ is calibrated using token-level sensitivity, so its numerical
value is not directly comparable across mechanisms.


\paragraph{Baseline handling.} We follow each baseline's intended usage while
adapting its inputs and preprocessing to our evaluation setting. SanText+ and
CusText require corpus-specific fitting: SanText+ derives its vocabulary and
sensitive-word set from the fitting corpus, while CusText derives the candidate
substitution set for each word. For the ai4privacy experiments
(Tables~\ref{tab:survives} and~\ref{tab:valueback}), we fit both methods on
ai4privacy rather than PUPA. Pr$\varepsilon\varepsilon$mpt's name-encryption
mechanism is designed to process multiple names jointly; when applied once per
name, it maps most names to the same surrogate. Therefore, for all ai4privacy
experiments (Tables~\ref{tab:pertype}, \ref{tab:e2e_stored},
\ref{tab:survives}, and~\ref{tab:valueback}), we first pass all distinct names
detected in the corpus in groups of at most $200$ and reuse the resulting
surrogates across prompts. For PUPA (Table~\ref{tab:crossed}), we instead run
Pr$\varepsilon\varepsilon$mpt once per prompt, following its per-request usage;
this changes the surrogate assigned to a name across prompts but does not
affect which original values are exposed.

\subsection{Evaluation Metrics}
\label{app:metrics}

Each metric is reported using the corresponding name of its table column. Throughout the tables, (\%) denotes a percentage, ($\times$) a ratio to the unprotected prompt, (\#) a count, and \% (\#) a percentage followed by its count.

\begin{itemize}[leftmargin=*]
\item {\bf \em Detection \emph{Recall}.} is the fraction of annotated entities for which the detector returns an overlapping entity.
\item {\bf \em \emph{Missed}.} is the number of annotated entities for which the detector returns no overlapping entity. Unless protected by another mechanism, these entities are sent to $\mathcal{M}$ verbatim.
\end{itemize}

\noindent {\bf \em Trust boundary metric categories.} Every method is evaluated on the text that reaches $\mathcal{M}$ using the following metrics:
\begin{itemize}[leftmargin=*]
\item {\bf \em \emph{Leaked}.} is the fraction of PII values that appear verbatim in the text sent to $\mathcal{M}$.
\item {\bf \em \emph{Leak-free}.} is the fraction of prompts containing no PII values verbatim in the text sent to $\mathcal{M}$.
\item {\bf \em \emph{Retained}.} is the fraction of the request's content words that remain in the text sent to $\mathcal{M}$. Leaked and Retained should be considered together because discarding more content can reduce leakage at the cost of utility.
\item {\bf \em \emph{Unchanged}.} is the fraction of prompts sent to $\mathcal{M}$ exactly as written by the user.
\item {\bf \em \emph{Words changed}.} is the mean fraction of words whose form or position differs in the transmitted text; unlike Retained, it counts reordered words as changed.
\item {\bf \em \emph{Perplexity (PPL)}.} is the median, across prompts, of the transmitted text's perplexity divided by that of the unprotected prompt. A value below one indicates more predictable text, not necessarily greater utility.
\item {\bf \em \emph{Local queries}.} is the number of queries to a local model required per request.
\end{itemize}

\noindent {\bf \em Transformation.} Transformation metrics are computed on detected entities using the following measures:
\begin{itemize}[leftmargin=*]
\item {\bf \em \emph{Protected}.} is the fraction of detected entities replaced by the method. A replacement that differs from the original only in case or spacing is not considered protected.
\item {\bf \em \emph{Round trip}.} is the fraction of protected entities that are restored exactly to their original values.
\item {\bf \em \emph{Valid}.} is the fraction of protected entities whose surrogates are valid instances of the corresponding entity type.
\item {\bf \em \emph{Collisions}.} is the number of distinct original values assigned the same surrogate, which prevents unambiguous restoration.
\end{itemize}
Detected dates are reported in two columns: \textsc{date} contains dates handled by the cipher, while \textsc{date fallback} contains the remaining dates, which \sys stores in the session vault when they contain a digit. All methods are evaluated on both sets. A cell marked n/a indicates that the method provides no mechanism for that entity type.

\vspace{2mm}
\noindent {\bf \em Date schemes.} Each date scheme is evaluated on dates containing at least six digits, which all three schemes can encipher: annotated dates from ai4privacy and dates matched by date patterns in WildChat and ShareGPT using the following measures:
\begin{itemize}[leftmargin=*]
\item {\bf \emph{Ambiguity kept}.} is the fraction of dates that parse as calendar dates for which the surrogate preserves exactly the same field-order interpretations as the original.
\item {\bf \emph{Cipher calls}.} is the mean number of cipher invocations required per date and is independent of machine speed.
\end{itemize}

\noindent {\bf \em Date translation.} Each ShareGPT document contributes its first two distinct dates that are at least two days apart, together with a stated present halfway between them. The untrusted model answers three questions under each scheme. The following measures are used:
\begin{itemize}[leftmargin=*]
\item {\bf \emph{Order}.} is the fraction of documents in which the model identifies the surrogate corresponding to the earlier original date as earlier.
\item {\bf \emph{Tense}.} is the fraction of documents in which the model correctly identifies both dates as past or future relative to the stated present.
\item {\bf \emph{Computed date}.} is the fraction of documents in which the date given as seven days after the second date, after applying the scheme's inverse, equals the original second date plus seven days.
\item {\bf \emph{Moved}.} is the median, across documents, of the smaller absolute movement in days between the two original dates and their corresponding surrogates, while \emph{Unmoved} is the number of documents in which neither date changes.
\end{itemize}

\noindent {\bf \em End to end.} End-to-end metrics evaluate whether surrogate values returned by the untrusted model can be restored locally and whether original values are ultimately recovered. The following measures are used:
\begin{itemize}[leftmargin=*]
\item {\bf \emph{Echoed}.} is the fraction of surrogates that the untrusted model returns verbatim in its answer.
\item {\bf \emph{Restored}.} is the fraction of echoed surrogates that are restored exactly to their original values.
\item {\bf \emph{Untrusted answer}.} is the fraction of values present in the untrusted model's answer, while \emph{Local answer} is the fraction present in the answer regenerated by a local model.
\item {\bf \emph{Value recovery}.} is the fraction of entities whose original value appears in the final answer, regardless of how the value is recovered.
\item {\bf \emph{Restored by}.} identifies the mechanism used to recover the original value: a session key without stored state, locally stored state such as a session vault or surrogate-to-value map, or a local model that rewrites the answer. A method may use multiple mechanisms for different entity types; \emph{not restored} denotes a type for which the method provides no recovery mechanism.
\end{itemize}






\subsection{Environment and Cryptographic Parameters}
\label{app:compute}

\paragraph{Computing environment.} All experiments were conducted on a single workstation running Ubuntu 24.04.1 LTS with an Intel Core i9-14900K CPU (24 cores, 32 threads), 125\,GB RAM, and an NVIDIA RTX PRO 6000 Blackwell GPU with 96\,GB memory (driver 575.57, CUDA 12.8). The implementation uses Python 3.12.4, PyTorch 2.10.0, and Transformers 4.57.6 for language-model inference; Presidio Analyzer 2.2.364 and spaCy 3.8.16 with the \texttt{en\_core\_web\_trf} 3.8.0 model for PII detection; and \texttt{cryptography} 48.0.1 for the AES, HKDF, and HMAC primitives.

\paragraph{Cryptographic parameters.} Each experimental run uses a 256-bit master key, from which 256-bit session keys are derived using HKDF-SHA256. FF1 uses AES-256 with ten Feistel rounds and radix~$10$, while the semantic mechanism uses HMAC-SHA256 for keyed hashing.

\end{document}